%% file: paper.tex
\documentclass{llncs}

\usepackage{amsmath}
\usepackage{amssymb}
\usepackage{booktabs}
\usepackage{colonequals}
\usepackage{tikz}
\usepackage{xspace}
\usepackage{marvosym}
\usepackage{xcolor}
\usepackage{wrapfig}
\usepackage{cite}
\usepackage{algpseudocode}
\usepackage{algorithmicx}
\usepackage{algorithm}
\usepackage{mathpartir}
\usepackage{multirow}
\usepackage{comment}

\usepackage{mathtools}

\usepackage{version}

\iftrue
\includeversion{arxiv}
\excludeversion{cav}
\else
\excludeversion{arxiv}
\includeversion{cav}
\fi

\algnewcommand{\IIf}[1]{\State\algorithmicif\ #1\ \algorithmicthen}
\algnewcommand{\EndIIf}{\unskip\ \algorithmicend\ \algorithmicif}

\usepackage{graphicx}                   
\usepackage{hyperref}                   
\usetikzlibrary{shapes,arrows,positioning,decorations.markings,matrix,fit,calc,arrows.meta,decorations.pathreplacing}
\tikzstyle{line} = [draw, -latex']
\pgfdeclarelayer{background}
\pgfdeclarelayer{foreground}
\pgfsetlayers{background,main,foreground}

\usepackage{hyperref}
\usepackage{cleveref}

\newcommand\paptitle{Satisfiability Modulo\\Extensional Constant Arrays}
\newcommand\papthanks{This work was supported in part by the Stanford Center
for Automated Reasoning, the Stanford Center for Blockchain Research,
Defense Advanced Research Projects Agency (DARPA) contract FA875024-2-1001,
and a gift from Amazon Web Services.}

\hypersetup{
 pdfborder={0 0 0},
 colorlinks=true,
 linkcolor=blue,
 urlcolor=blue,
 citecolor=blue,
 pdftitle={\paptitle}
}
\makeatletter
\def\orcidID#1{\smash{\href{http://orcid.org/#1}{\protect\raisebox{-1.25pt}{\protect\includegraphics{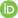}}}}}
\makeatother

\begin{document}

\include{macros.tex}

\title{\paptitle}

\author{%
  Mathias Preiner\textsuperscript{(\Letter)}\hskip .1em\orcidID{0000-0002-7142-6258}
  \and
  Aina Niemetz\hskip .1em\orcidID{0000-0003-2600-5283}
  \and
  Clark Barrett\hskip .1em\orcidID{0000-0002-9522-3084}
}
\authorrunning{M. Preiner et al.}
\institute{Stanford University, Stanford, USA\\\email{\{preiner,niemetz,barrett\}@cs.stanford.edu}}

\maketitle

\begin{abstract}
  Reasoning about array data structures is a key requirement for
  many applications in hardware and software verification, especially in
  combination with machine integers.
  The Satisfiability Modulo Theories (SMT) theory of extensional arrays
  provides array read and write operators and allows extensionality over
  arrays.
  This is sufficient to express many aspects of computer-aided verification,
  but lacks succinctness to efficiently deal with arrays that are initialized
  with a default value.
  Existing procedures for extending the SMT-LIB theory of arrays with support
  for constant arrays are limited to arrays with infinite index domains,
  and existing implementations in SMT solvers only support a fragment of the
  theory for finite index domains.
  In this paper, we present a novel decision procedure for the theory of
  arrays with constant arrays that supports arbitrary index domains and is not
  limited to the infinite case.
  We present our procedure as an abstract calculus and show its
  refutational and satisfiability soundness.
  We implement a decision procedure based on our calculus in the
  state-of-the-art SMT solver Bitwuzla and evaluate its performance on a
  diverse collection of benchmarks and use cases.
\end{abstract}

\section{Introduction}

Arrays are an essential data structure in programming languages, and are
commonly used to model memory in hardware.
The ability to efficiently reason about arrays, especially in combination with
machine integers, is therefore a core requirement for many applications
of computer-aided verification~\cite{btor2,handbookmc18,BaldoniCDDF18}.
Satisfiability Modulo Theories (SMT) solvers for the theory of arrays provide
such reasoning and serve as the back end engine for a wide range of
applications.

The SMT-LIB standard~\cite{BarFT-RR-25} defines a theory of arrays with
extensionality as originally axiomatized by McCarthy~\cite{McCarthy62}.
This theory provides array read (\texttt{select}) and write (\texttt{store})
operators and allows extensionality (equality) over arrays.
While this is sufficient to express many aspects of modeling memory and array
data structures in hardware and software,
the theory lacks the ability to succinctly express operations over multiple
indices and memory initialization.
When modeling an array initialized with a default value,
e.g., to represent a zero-initialized chunk of memory,
this initialization must be expressed via a nested store expression that
explicitly writes the default value to all relevant indices.
This leads to inaccuracies in the modeling as indices can only be
explicitly overwritten, and to inefficiencies in solving procedures as the
nesting depth of such store chains can become very deep, depending on the
number of relevant indices.

Extending the SMT-LIB theory of arrays with so-called \emph{constant
arrays}, i.e., arrays that store a single element at every
index, provides a succinct way to model array initialization with a default
value (which may be an arbitrary term).
Existing procedures for this extended theory of \emph{extensional} constant
arrays, however, are all restricted to
\emph{infinite} index domains~\cite{StumpBDL01,HoenickeS19,MouraB09}.

In this paper, we present a novel decision procedure
for the \emph{theory of extensional constant arrays}
that is \emph{not} restricted to
infinite index sorts.
Our procedure is presented as an abstract calculus and is a generalization and
extension of the \emph{Lemmas on Demand (LOD)} procedure for the extensional
theory of arrays presented in~\cite{BrummayerB09}.
In particular, we make the following \emph{contributions}:

\begin{itemize}
  \item We generalize the LOD procedure for the extensional
    theory of arrays with fixed-size bit-vectors presented
    in~\cite{BrummayerB09} to arbitrary index and element sorts.
  \item We present this generalization as an abstract calculus
    \calcbase, which serves as the core calculus of our new decision procedure
    and is \emph{decoupled} from the LOD solving paradigm~\cite{DemouraR02,BarrettDS02}.
  \item We then extend this core calculus with support for extensional constant
    arrays, and prove refutational and satisfiability soundness of our resulting
    calculus \calcext for the theory of extensional constant arrays.
  \item We implemented our procedure in the state-of-the-art SMT solver
    Bitwuzla~\cite{bitwuzla}. We evaluate it on a diverse
    collection of benchmarks and use cases,
    including applications in hardware model checking
    and smart contract verification of the Ethereum virtual machine.
    Our evaluation shows that our new procedure significantly expands the
    solving capabilities of Bitwuzla.
\end{itemize}

\paragraph{Related Work.}
The extension of the theory of arrays with constant arrays is not
yet standardized in SMT-LIB~\cite{BarFT-RR-25}.
Over the years, several array procedures incorporated support for constant
arrays to some extent, with limitations~\cite{StumpBDL01,MouraB09,HoenickeS19,PreinerNB13}.
Several state-of-the-art SMT
solvers~\cite{z3,cvc5,btor2,bitwuzla,smtinterpol,yices2,mathsat5}
have limited support for
constant arrays.
To the best of our knowledge,
none provide \emph{full} support for extensional constant arrays with
\emph{finite} index domains.

In 2001, Stump et al.~\cite{StumpBDL01} axiomatized the
extensional theory of arrays with constant arrays and
presented a sound and complete decision procedure for the
extensional theory of arrays. The authors also proposed
an extension of the procedure with support for constant arrays over \emph{infinite} index domains,
which is not implemented in any state-of-the-art SMT solver.

In 2009,
de Moura and Bj{\o}rner~\cite{MouraB09}
proposed a procedure for the so-called combinatory array logic
(CAL) extension of the extensional array theory,
which is implemented by the SMT solver Z3~\cite{z3}.
The CAL extension includes constant arrays,
and the procedure is, again, limited to \emph{infinite} index domains.
While Z3 accepts problems that use constant arrays with finite index domains,
in our experiments, we encountered cases where it reported wrong answers.
All of these instances involved extensionality over constant arrays with finite
index domains.
We faced a similar issue with MathSAT5~\cite{mathsat5}, which
supports constant arrays and accepts input problems over extensional constant arrays
with finite index domains but reported wrong answers for some instances
involving extensionality.

cvc5~\cite{cvc5} implements a procedure based on~\cite{MouraB09},
restricting support for constant arrays to the non-extensional case,
and to interpreted constants as elements (arbitrary terms are not allowed).
SMTInterpol~\cite{smtinterpol} supports satisfiability and interpolation
queries for the extensional theory of arrays with constant arrays as presented
by Hoenicke and Schindler~\cite{HoenickeS19} in 2019, also limited to the
\emph{infinite} index domain case.
Yices2~\cite{yices2} does not have native support for constant arrays, but has
limited support for non-extensional constant arrays by means of encoding them
as lambda terms. However, Yices2 gives up if the lambda terms cannot be
eliminated via beta reduction.
Boolector~\cite{btor2} generalizes the lemmas on demand
procedure for extensional arrays described
in~\cite{BrummayerB09} to non-recursive first-order lambda
terms~\cite{PreinerNB13,preiner-phd}
and allows reasoning about constant arrays
(modeled as lambda terms) for the \emph{non-extensional} theory of arrays with
constant arrays.

Bitwuzla~\cite{bitwuzla} implements and extends the array procedure
from~\cite{BrummayerB09} with support for
(extensional) nested arrays and \emph{non-extensional} constant arrays.
Until now, extensionality over constant arrays was not supported.

\section{Preliminaries}

We assume and briefly review the usual notions and terminology of many-sorted
first-order logic with equality
(see, e.g.,~\cite{EndertonLogic,Man-MSL-93} for a more
detailed introduction).

Let \sorts be a set of \emph{sort symbols},
and let \sig be a \emph{signature} containing
a set $\sigs\!\subseteq \sorts$ of sort symbols
and a set \sigf of function symbols
$f^{\sort_1 \cdots \sort_n \sort}$
with arity $n \ge 0$ and $\sort_1, ..., \sort_n, \sort \in \sigs$.
We usually omit the superscript from function symbols
and refer to 0-arity function symbols as \emph{constants}.

We assume that \sig includes a designated sort \bool,
Boolean constants \btrue~(true) and \bfalse (false),
Boolean connectives \{\boolandf, \boolnotf\} defined as usual,
equality and disequality symbols $\{\teqf, \tneqf\}$ of sort
$\sort\times\sort\to\bool$ for every $\sort\in\sigs$,
and an if-then-else operator
\emph{ite} of sort $\bool\times\sort\times\sort\to\sort$ for every
$\sort\in\sigs$.

Let \I be a \emph{\sig-interpretation} that maps
each $\sort \in \sigs$ to a non-empty set \sorti (the \emph{domain} of \I),
with $\bool^\I = \{ \btrue, \bfalse \}$;
and each $f^{\sort_1 \cdots \sort_n \sort} \in \sigf$ to a total function
$f^\I\!\!: \sort_1^\I \times ... \times \sort_n^\I \to \sorti$ if $n > 0$,
and to an element in \sorti if ${n = 0}$.
We denote the domain size of a sort $s \in \sigs$ as~\domsize{s}, e.g.,
$\domsize{\bool} = 2$.
The interpretation of Boolean connectives, Boolean constants,
equality symbols and \emph{ite} symbols
is fixed and standard.
%
We use the usual inductive definition of the satisfiability relation $\models$
between \sig-interpretations and \sig-formulas.
%

A \emph{theory} is a pair $(\sig,\Is)$ where $\sig$ is some signature,
and $\Is$ is a class of $\sig$-interpretations.
A \sig-formula is \emph{\T-satisfiable} (resp.~\emph{\T-unsatisfiable})
if it is satisfied by some (resp.~no) interpretation in \Is;
it is \emph{\T-valid} if it is satisfied by all interpretations in \Is.
We assume the usual definition of well-sorted terms, literals, and formulas,
and call $\sig$-formulas $\T$-formulas and $\sig$-literals
$\T$-literals.
We sometimes refer to standard first-order logic as the \emph{empty} theory \theoryeuf.

A literal is \emph{flat} if it is of the form
\bfalse,
$p(x_1,\ldots,x_n)$,
$\neg p(x_1,\ldots,x_n)$,
\teq{x}{y},
\tneq{x}{y},
or
$\teq{x}{f(x_1,\ldots,x_n)}$,
where $p$ and $f$ are function symbols and $x$, $y$, and $x_1,\ldots,x_n$ are
uninterpreted constants.

\section{A Theory of Extensional Constant Arrays}

The SMT-LIB theory of extensional arrays \theoryarr
defines a parameterized sort \arraysort for arrays with index sort \indexsort
and element sort \elementsort.
Its signature contains two interpreted function symbols for accessing
(\texttt{select}) and writing (\texttt{store}) to arrays
as axiomatized by McCarthy in~\cite{McCarthy62}.

In the following,
we use symbol \indexsort for array index sorts,
symbol \elementsort for array element sorts,
symbols $a$, $b$, $c$, $d$, $e$ for arrays of sort \arraysort,
symbols $i$, $j$, $k$ for indices of sort \indexsort,
and symbols $u$, $v$, $w$ for elements of sort \elementsort.

We define the \emph{theory of extensional constant arrays} $\theorycarr = (\sigcarr, \Icarr)$ as an extension
of \theoryarr with an additional interpreted function symbol for constructing
constant arrays.
The signature~\sigcarr of theory \theorycarr is shown in \Cref{tab:arrops}.
Array \emph{access} function \aread{a}{i}
reads an element from array~$a$ at index~$i$ (\texttt{select}),
and array \emph{update} function \awrite{a}{i}{u} constructs a new
array that is identical to array~$a$ except for index $i$ where it stores
element~$u$ (\texttt{store}).
Further, as originally proposed by Stump et al.~in~\cite{StumpBDL01},
\sigcarr includes a constant array \emph{constructor} operator~$\aconsti{v}$.
This operator creates an array of sort \arraysort
that stores a single element~$v$ at all indices and
is currently not defined in SMT-LIB.
The majority of SMT solvers that support constant arrays,
however, support the \texttt{((as const (Array IndexSort ElementSort)) v)} SMT-LIB
construct as a constant array constructor.
We omit $\indexsort$ from  $\aconsti{v}$ if the index sort is clear
from the context.

\begin{table}[t]
  \centering
  \caption{Signature for the theory of constant arrays \theorycarr.}
  \input{table/arrayops.tex}

  \label{tab:arrops}
\end{table}

\subsection{Axiomatization of \theorycarr}

The theory of extensional constant arrays \theorycarr is axiomatized by
the following four axioms.
Array access and update are defined by McCarthy's
\emph{read-over-write} axioms \ref{eq:roweq} and \ref{eq:rowne}~\cite{McCarthy62}.
Note that the axiomatization scheme for the non-extensional theory of arrays
only includes these two axioms.

Axiom \ref{eq:roweq} states that if array \awrite{a}{i}{u}, which results from
storing element~$u$ at index $i$ in array $a$, is accessed at that index,
the accessed element must be $u$.
\begin{equation}\tag{row-eq}
  \forall\,
    a,i,j,u
    \,.\, \imp{\teq{i}{j}}{\teq{\aread{\awrite{a}{i}{u}}{j}}{u}}
  \label{eq:roweq}
\end{equation}
Axiom \ref{eq:rowne} covers the symmetric opposite case and states that all
elements that were not overwritten by store \awrite{a}{i}{u} remain unchanged.
\begin{equation}\tag{row-ne}
  \forall\,
  a,i,j,u
  \,.\, \imp{\tneq{i}{j}}{\teq{\aread{\awrite{a}{i}{u}}{j}}{\aread{a}{j}}}
  \label{eq:rowne}
\end{equation}

The \emph{extensionality} axiom~\ref{eq:ext} states that two arrays are equal
if and only if both store the same element on each index.
The axiomatization scheme for the extensional theory of arrays \theoryarr
consists of this and the previous two axioms.
%
\begin{equation}\tag{ext}
  \forall\,
    a,b
    \,.\, \bimp{\teq{a}{b}}{\forall\, i \,.\, \teq{\aread{a}{i}}{\aread{b}{i}}}
  \label{eq:ext}
\end{equation}
The axiomatization scheme for the extensional theory of constant arrays
\theorycarr
extends the scheme for \theoryarr with
the \emph{read-over-constant-array} axiom~\ref{eq:roc}~\cite{StumpBDL01}.
This axiom states that accessing a constant array at any arbitrary index always
yields its default value (the value it was constructed with).
%
\begin{equation}\tag{roc}
  \forall\,
    i,v
    \,.\, \teq{\aread{\aconst{v}}{i}}{v}
  \label{eq:roc}
\end{equation}

\subsection{Infinite vs.~Finite Constant Arrays}

A major limitation of previous work on extensionality of constant
arrays~\cite{StumpBDL01,HoenickeS19,MouraB09} is its restriction to
\emph{infinite} index domains.
In the following, we illustrate
why the restriction to infinite constant arrays significantly simplifies
procedures for extensional constant arrays.

Given two constant arrays $\aconst{v}$ and $\aconst{w}$,
from axioms \ref{eq:ext} and \ref{eq:roc}
it follows that \aconst{v} and \aconst{w} are
equal if and only if their default values $v$ and $w$ are equal,
i.e., $\forall v, w\,.\, \bimp{\teq{\aconst{v}}{\aconst{w}}}{\teq{v}{w}}$.
For constant arrays with \emph{infinite index domain},
the same also holds
for equality over array stores on constant arrays, e.g.,
\teq{\awrite{\awrite{\awrite{\aconst{v}}{i}{w}}{j}{w}}{k}{w}}{\aconst{w}}.
Intuitively,
infinite constant arrays with different default values cannot be ``made'' equal
since this requires an \emph{infinite} number of array stores,
which cannot be expressed in the array theory.
Thus, for infinite constant arrays \aconst{v} and \aconst{w}, an equality over update sequences on
\aconst{v} and \aconst{w} cannot hold if \teq{v}{w} does not hold.

For \emph{finite index domains}, however, two constant arrays with different
default values can be made equal if sufficiently many stores ensure that they
are equal on every index.
For example, formula
$\teq{\awrite{\awrite{\aconst{v}}{i}{w}}{j}{w}}{\aconst{w}} \wedge \tneq{v}{w}$
is satisfiable
for constant arrays with an index domain of size two,
since it is possible to find an interpretation for $i$ and $j$
such that all indices of \aconst{v} are updated to $w$.

More generally, two finite constant arrays \aconsti{v}
and \aconsti{w} with different default values $v$ and $w$
differ on a finite number of indices $n = \domsize{\indexsort}$.
Consider the following equality over array update sequences
on \aconst{v} and \aconst{w}:
\begin{align*}
  \teq{\awrite{\awrite{\awrite{\aconst{v}}{i_1}{w}}{\ldots}{\ldots}}{i_l}{w}}
      {\awrite{\awrite{\awrite{\aconst{w}}{j_1}{v}}{\ldots}{\ldots}}{j_k}{v}}
\end{align*}

If \tneq{v}{w},
this equality only holds if
the indices \emph{not} updated in array update sequence
$s_v \coloneq \awrite{\awrite{\awrite{\aconst{v}}{i_1}{w}}{\ldots}{\ldots}}{i_l}{w}$
are updated in
$s_w \coloneq \awrite{\awrite{\awrite{\aconst{w}}{j_1}{v}}{\ldots}{\ldots}}{j_k}{v}$,
i.e., $s_v$ and $s_w$ must perform at least $n$ array updates,
on $n$ distinct indices.
In other words,
interpretation \I satisfies formula $\teq{s_v}{s_w} \wedge \tneq{v}{w}$
if and only if $\indexsort^\I = \{i_1^\I, \ldots, i_l^\I, j_1^\I, \ldots, j_k^\I\}$.
For any number of array stores $l + k < n$, the equality does not hold
since $s_v$ and $s_w$ differ on at least $n - (l + k)$ indices.

Thus, the main challenge of deciding the equality of two store sequences over
different finite constant arrays is to check whether there exists an
arrangement of $n\!=\!\domsize{\indexsort}$ unique array updates such that the resulting
updated arrays are equal.

\begin{example}
  \label{ex:infinite}
  Consider formula $\varphi$ over \emph{infinite} constant arrays,
  and assume that arrays \aconst{v}, $a$, and \aconst{w} are indexed by
  integers and store Booleans as elements.
  \begin{align*}
    \varphi \coloneq
    \teq{\awrite{\aconst{v}}{i_1}{u_1}}{\awrite{a}{j_1}{u_2}}
    \wedge
    \teq{\awrite{a}{i_2}{u_3}}{\awrite{\aconst{w}}{j_2}{u_4}}
  \end{align*}
  Formula $\varphi \wedge \teq{v}{w}$ is \emph{satisfiable} and satisfied, e.g., by interpretation
  \begin{align*}
  \{
    a = \aconst{\bfalse},
    v = w = u_1 = u_2 = u_3 = u_4 = \bfalse,
    i_1 = i_2 = j_1 = j_2 = 0
  \}
  .
  \end{align*}
  Formula $\varphi \wedge \tneq{v}{w}$, however, is \emph{unsatisfiable} since
  \aconst{v} and \aconst{w} store different elements at infinitely many
  indices.
  Thus,
  there exists no arrangement for array updates at indices
  $i_1$, $i_2$, $j_1$, and $j_2$ to make the array update sequences equal.
  \qed
\end{example}

\begin{figure}[t]
  \centering
  \resizebox{\textwidth}{!}{
    \input{fig/example2.tex}
  }
  \caption{Visualization of the interpretation satisfying formula
           $\varphi \wedge \tneq{v}{w}$ from \Cref{ex:finite}.}
  \label{fig:example_dag}
\end{figure}
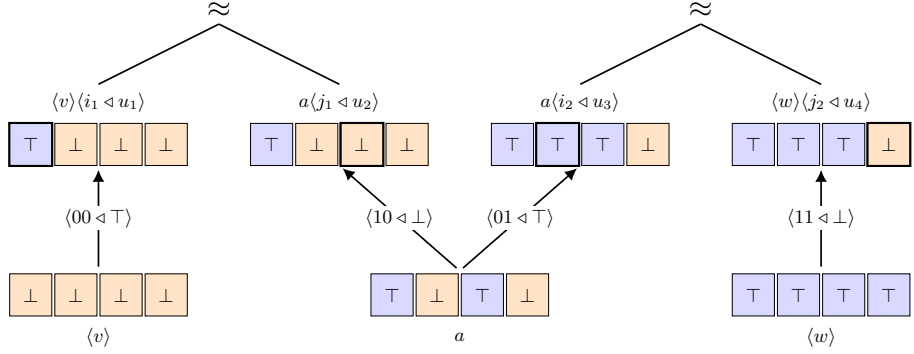

\begin{example}
  \label{ex:finite}
  Now, for formula $\varphi$ from \Cref{ex:infinite} above, assume that
  arrays~\aconst{v}, $a$, and~\aconst{w} are indexed by
  bit-vectors of size two, i.e., by a \emph{finite} index sort.
  Formula $\varphi \wedge \teq{v}{w}$ is still \emph{satisfiable} and satisfied, e.g.,
  by interpretation
  \begin{align*}
  \{
    a = \aconst{\bfalse},
    v = w = u_1 = u_2 = u_3 = u_4 = \bfalse,
    i_1 = i_2 = j_1 = j_2 = 00
  \}
  .
  \end{align*}
  Interestingly, formula $\varphi \wedge \tneq{v}{w}$ is now also \emph{satisfiable} since
  the array stores on indices $i_1$, $i_2$, $j_1$, and $j_2$ can be arranged in
  a way such that each equality holds.
  One possible interpretation satisfying this formula,
  as illustrated by \Cref{fig:example_dag},
  is
  \begin{align*}
  \{
    a &= \{
      00 \mapsto \btrue,
          01 \mapsto \bfalse,
          10 \mapsto \btrue,
          11 \mapsto \bfalse
        \},\\
    v &= u_2 = u_4 = \bfalse,
    w = u_1 = u_3 = \btrue,\\
    i_1 &= 00,
    i_2 = 01,
    j_1 = 10,
    j_2 = 11
  \}
  .
  \end{align*}

  On the other hand, formula $\varphi \wedge \tneq{v}{w} \wedge \teq{i_1}{j_1}$ is
  \emph{unsatisfiable} since constraint $\teq{i_1}{j_1}$ decreases the number of updates on
  distinct indices to three,
  but $\domsize{\indexsort} = 4$ unique updates are required to satisfy
  both equalities.
  Similarly, formula
  $\varphi \wedge \tneq{v}{w}\wedge \allowbreak \tneq{i_1}{j_1} \wedge \teq{u_1}{u_2}$
  is also unsatisfiable.
  While the array stores can be arranged to update four distinct indices, the
  additional constraint \teq{u_1}{u_2} results in array store
  \awrite{\aconst{v}}{i_1}{u_1} to update index~$i_1$ to $u_2$.
  However, since \tneq{i_1}{j_1}, we have that \teq{u_2}{v} and thus,
  the array update on $i_1$ leaves \aconst{v} unchanged while the remaining three
  array updates are not sufficient to satisfy both equalities.
  \qed
\end{example}

\section{Calculus}
\label{sec:calc}

In this section, we first generalize and formalize the decision procedure for
the extensional theory of arrays presented in~\cite{BrummayerB09}, which we
refer to as \dpa, as an abstract base calculus \calcbase.
We then extend \calcbase with a set of new rules
to handle extensional finite and infinite constant arrays,
yielding a new calculus \calcext,
which maintains full
generality with respect to the finiteness of the index domain of constant arrays.
Finally, we prove that \calcext is sound and complete for the theory of
extensional constant arrays \theorycarr.

In the following, we briefly introduce common assumptions, definitions and
notation used to define our base calculus and its extension.

\begin{assumption}
  Whenever we refer to a set of \sigarr-formulas \assertions,
  we assume that all literals in \assertions are \emph{flat}.
  \qed
\end{assumption}

This assumption is without loss of generality as any set of \sigarr-formulas
can be easily transformed into a \theoryarr-equisatisfiable set by the
addition of fresh variables and equalities.
Note that some of the derivation rules we introduce later
introduce non-flat literals in \assertions.
In such cases, we assume that similar transformations
maintain the invariant that all literals in \assertions are flat.

\begin{definition}[Configuration]
  A configuration is either the distinguished configuration \unsat or
  a tuple \config{\assertions}{\I}{\propathf}, where
  \assertions is a set of \sigarr-formulas,
  \I is an interpretation in the empty theory \theoryeuf that satisfies
  \assertions, and \propathf is a map of propagation steps
  (as defined in \Cref{def:propstep}).
  \qed
\end{definition}

Calculus rules modify configurations and are given in
\emph{guarded assignment form}, where the rule premises describe the conditions
on the current configuration under which the rule can be applied.
The \emph{conclusion} of a rule is either \unsat or otherwise describes the
resulting modifications to the configuration.
An application of a rule is \emph{redundant} if it has a conclusion where each
component of the derived configuration is equal to the corresponding
component in the premise configuration, i.e., it remains unchanged.
A configuration is \emph{saturated} with respect to a set of rules if every possible
application of each rule is redundant.

In the following,
we use \terms{\assertions} to denote the set of all terms in
formulas~\assertions, and \termsa{\assertions} to denote all array terms
in~\assertions.
Further, we assume that array disequalities
$\tneq{a}{c} \in \assertions$ are represented as
$\neg(\teq{a}{c})$.

\subsection{Base Calculus}
\label{sec:base}

In this section, we generalize and formalize the 
LOD decision
procedure \dpa for the extensional theory of arrays of~\cite{BrummayerB09} as
our base calculus~\calcbase.

The main concept of \dpa as presented in~\cite{BrummayerB09}
is based on a fixed point algorithm for
propagating and congruence checking array reads.
The algorithm produces lemmas that are derived from McCarthy's array axioms
\emph{without} introducing new array reads.
In~\cite{BrummayerB09}, procedure \dpa assumes that \theoryarr is combined with
a single, decidable theory for array indices and elements,
and is presented in the context of combining
\theoryarr with the theory of fixed-size bit-vectors \theorybv.
It is based on the lemmas on demand paradigm~\cite{DemouraR02,BarrettDS02},
but with a bit-vector abstraction at its core,
and relies on necessary formula transformations during preprocessing
that introduce virtual reads for array equalities and writes.
The authors show that for this combination of theories, \dpa is sound and
complete.

We generalize \dpa to combinations of \theoryarr with arbitrary theories for
array indices and elements and decouple it from the LOD paradigm.
Our generalized procedure for the theory of extensional arrays \theoryarr
does not rely on necessary formula transformations during preprocessing,
does not assume that array index and element sorts belong to the same theory,
and does not rely on the existence of a decision procedure for index and
element theories.
Instead, it assumes an interpretation of index and element terms in the empty
theory \theoryeuf.

We formalize our generalized procedure in terms of our \emph{base calculus} \calcbase.
Given a set of \sigarr-formulas \assertions and a current interpretation
$\I \models \assertions$,
the derivation rules of our base calculus
aim to \emph{propagate information} about array indices and elements from one array
to another, based on the axioms of \theoryarr.
For example, for an array access \aread{\awrite{a}{j}{u}}{i},
if $\I \models \tneq{i}{j}$, from axiom~\ref{eq:rowne} it follows that
\teq{\aread{\awrite{a}{j}{u}}{i}}{\aread{a}{i}}.
We can therefore propagate the information that reading from \awrite{a}{j}{u}
at index $i$
yields the same element as reading from~$a$ at index $i$.
Our calculus explores propagations under the current interpretation
until unsatisfiability is derived, or all possible propagations have been explored.

Calculus \calcbase consists of the set of derivation rules given in
\Cref{fig:calcbasene,fig:calcbase}.
\Cref{fig:calcbasene} lists all rules required for the \emph{non-extensional}
theory of arrays, while \Cref{fig:calcbase} shows the supplemental
rules to handle the \emph{extensional} theory of arrays.

\begin{definition}[Propagation Step]
  A propagation step $\propath{a}{t} \coloneq (r,c)$
  maps a propagated term~$t$ and its destination array $a$ to a tuple
  $(r, c)$,
  where $r$ defines the reason for propagating $t$ from source array
  $c$ to the destination array $a$.
  \qed
  \label{def:propstep}
\end{definition}

Given a configuration \config{\assertions}{\I}{\propathf},
the \emph{propagation path} of a read \aread{b}{i} to destination array $a$
can be constructed as a sequence of propagation steps
by following the steps in reverse order.
That is, starting from \propath{a}{\aread{b}{i}},
we reach \propath{b}{\aread{b}{i}} via the propagation path
$\{
    \propath{a}{\aread{b}{i}} \coloneq (r_n,c_n),
    \propath{c_n}{\aread{b}{i}} \coloneq (r_{n-1},c_{n-1}),
    \ldots,\allowbreak
    \propath{c_1}{\aread{b}{i}} \coloneq (r_1,b),
    \propath{b}{\aread{b}{i}} \coloneq (\btrue,b)
\}$.
The reason \reasonof{a}{\aread{b}{i}} for propagating~\aread{b}{i} from $b$ to
$a$ along this path is defined as follows.

\begin{definition}[Propagation Reason]
  For a propagation step $\propath{a}{t} \coloneq (r,c)$,
  \reasonof{a}{t} defines the reason for propagating $t$
  from its initial source array to destination array $a$ as follows.
  \begin{equation*}
  \reasonof{a}{t} =
  \begin{cases}
    \emptypair               &\text{ if } \propath{a}{t} = \emptypair\\
    \btrue                   &\text{ if } t = a \text{ or } t = a[i] \text{ for some } i\\
    \reasonof{c}{t} \wedge r &\text{ otherwise, where } \propath{a}{t} = (r,c)
  \end{cases}
  \end{equation*}
  \qed
  \label{def:reason}
\end{definition}

Note that in the context of our base calculus \calcbase,
we only propagate reads~\aread{b}{i} to destination arrays $a$, i.e.,
propagated term $t$  in \Cref{def:propstep,def:reason} is always a read.
Thus, the second case of \Cref{def:reason} for $t = a$ does not apply and
propagation reasons are only defined as \reasonof{a}{\aread{b}{i}}.
When extending \calcbase with support for constant arrays to \calcext,
we also propagate default values of constant arrays which we represent with
case $t = a $, as we will explain later.

\begin{figure}
  \centering
  \setlength{\fboxsep}{10pt}
  \fbox{\scalebox{0.9}{\input{rules/non_ext.tex}}}
  \caption{\calcbase rules for non-extensional arrays.}
  \label{fig:calcbasene}
  \bigskip
  \bigskip
  \fbox{\scalebox{0.9}{\input{rules/ext.tex}}}
  \caption{\calcbase rules for extensional arrays.}
 \label{fig:calcbase}
\end{figure}

\begin{figure}[t]
  \centering
  \setlength{\fboxsep}{10pt}
  \fbox{\scalebox{0.9}{\input{rules/cnon_ext.tex}}}
  \caption{\calcext rule for non-extensional constant arrays.}
  \label{fig:calcextne}
  \bigskip
  \bigskip
  \fbox{\scalebox{.88}{\input{rules/cext.tex}}}
  \caption{\calcext rules for extensional constant arrays.}
  \label{fig:calcext}
\end{figure}

\paragraph{Start.}
Our calculus starts with the \emph{initial} configuration
\config{\assertions}{\I_0}{\propathf_0}
with an empty interpretation $\I_0 \coloneq \none$
and an empty map of propagation steps $\propathf_0 \coloneq \emptymap$.
From this initial configuration,
it explores possible propagation steps \propathf under a \emph{current
interpretation} \I for formulas \assertions in the empty theory \theoryeuf, as
defined by rule \ruleinterp, while modifying the configuration.
If \assertions is \theoryeuf-unsatisfiable,
rule \ruleconf derives configuration \unsat.
Else, if the configuration is \emph{saturated}, i.e., each rule application is
redundant and configuration \rn{unsat} was not derived,
for each array $a \in \termsa{\assertions}$,
propagation steps \propath{a}{\cdot} represent a partial model of~$a$, which
can be extended to an interpretation~$\I_\mathcal{A}$
that satisfies \sigarr-formulas~\assertions.

\paragraph{Propagation.}
In addition to rules \ruleinterp and \ruleconf, we distinguish
\emph{propagation rules} that modify the map of propagation steps \propathf
(\rn{Init*}, \rn{Row*}, \rn{Eq*}), and \emph{conflict rules} that modify
the set of formulas \assertions (\rulecongr, \rulediseq).

Rules \ruleinitr and \ruleinitw initialize map \propathf for each
\emph{direct read} $\aread{a}{i} \in \terms{\assertions}$ and
\emph{virtual read} $\aread{\awrite{a}{i}{u}}{i}$
for each write $\awrite{a}{i}{u} \in \terms{\assertions}$, respectively.
The purpose of these virtual reads is to represent the written element $u$
in terms of a read \aread{\awrite{a}{i}{u}}{i}, which allows for uniform
handling of array read and write terms in the calculus.
This is similar to the approach of \dpa, where
virtual reads are conceptually added to the formula via top-level equalities
\teq{\aread{\awrite{a}{i}{u}}{i}}{u} as a preprocessing step.
In our calculus, however, the congruence of \aread{\awrite{a}{i}{u}}{i} and $u$
is ensured by rules \ruleinterp and \ruleconf.

Rule \rulerowd is based on a rule of \dpa which applies read-over-write axiom~\ref{eq:rowne}.
Given a read \aread{\awrite{a}{j}{u}}{i}, if \tneq{i}{j},
by axiom~\ref{eq:rowne} we have that
\teq{\aread{\awrite{a}{j}{u}}{i}}{\aread{a}{i}}.
However, instead of creating a (potentially) fresh read \aread{a}{i},
we can use the existing read \aread{\awrite{a}{j}{u}}{i} as
representative of \aread{a}{i} under the current interpretation~\I, and
\emph{propagate} \aread{\awrite{a}{j}{u}}{i} from source array
\awrite{a}{j}{u} down to destination $a$, which is recorded in the conclusion of
rule \rulerowd by adding propagation step
$\propath{a}{\aread{\awrite{a}{j}{u}}{i}} \coloneq (\tneq{i}{j}, \awrite{a}{j}{u})$ 
to the configuration.

\paragraph{Conflict.}
Rule \rulecongr is based on the congruence rule of \dpa, which handles the
case where two reads \aread{b}{i} and \aread{c}{k} read different elements
from the same array~$a$ at the same index.
Recall that
with
$\propath{a}{\aread{b}{i}} \ne \emptypair$
and
$\propath{a}{\aread{c}{k}} \ne \emptypair$,
we have that
\aread{b}{i} is a representative of \aread{a}{i} and
\aread{c}{k} is a representative of \aread{a}{k} under the current
interpretation \I.
However, if
$\I \models \teq{i}{k} \wedge \tneq{\aread{b}{i}}{\aread{c}{k}}$,
the congruence axiom of \theoryeuf is violated since
$
\teq{i}{k}
\imp{\wedge \teq{\aread{b}{i}}{\aread{a}{i}} \wedge \teq{\aread{c}{k}}{\aread{a}{k}}}
    {\teq{\aread{b}{i}}{\aread{c}{k}}}
$.
As a result, \rulecongr adds a lemma to \assertions to rule out the current
spurious interpretation~\I.
Note that by modifying \assertions, configuration components \propathf and \I are reset to
their initial values $\propathf_0$ and $\I_0$ since \propathf is dependent on
the current interpretation \I, which does not satisfy the modified formulas
\assertions anymore.

\paragraph{Extensionality.}
The derivation rules given in \Cref{fig:calcbase} are only required
if \assertions contains extensionality over array terms.
Similarly to rule \rulerowd,
if $\I \models \tneq{i}{j}$,
rule \rulerowu applies axiom~\ref{eq:rowne} but propagates \aread{a}{i} to
\awrite{a}{j}{u}, since the element of~$a$ at index $i$ remains unchanged after
the update operation.
Rules \ruleeqr and \ruleeql apply the extensionality axiom \ref{eq:ext} and
propagate reads on arrays $a$ and $c$ over equalities
$\teq{a}{c} \in \terms{\assertions}$ from left to right and vice versa.

For array equalities \teq{a}{c} for which $\I \models \tneq{a}{c}$,
rule \rulediseq adds an instance of axiom~\ref{eq:ext} to \assertions.
This rule is the only rule that introduces fresh reads (on a fresh index
\adiff{a}{c}),
however, limited to at most two reads per equality.
The added lemma ensures that there exists
at least one witness index \adiff{a}{c} at which arrays $a$ and $c$ differ.
In contrast, \dpa relies on a preprocessing step that eagerly adds this lemma
for each $\teq{a}{c} \in \terms{\assertions}$, independent of whether $\I \models \tneq{a}{c}$.

\subsection{Constant Array Extension}
\label{sec:ext}

In this section, we extend our base calculus \calcbase with a set of new rules
to support extensional constant arrays.
The new rules for our extended calculus~\calcext are given in
\Cref{fig:calcextne,fig:calcext}.
\Cref{fig:calcextne} shows the additional conflict rule \ruleroc required for
handling reads over constant arrays, while \Cref{fig:calcext}
lists the new rules to handle extensional constant arrays.

The main challenge of determining satisfiability of equalities between store sequences over
different finite constant arrays is that we have to reason about default
values for indices that have not been updated by array stores.
A straightforward extension of \calcbase would be the introduction of virtual reads
for all indices of a constant array, i.e., adding
$\teq{\aread{\aconsti{v}}{i}}{v}$ for all $i \in \indexsort^\I$.
However, this is not feasible in general,
neither for infinite nor for finite index domains.
For the infinite case, we would never reach a saturated configuration due to
infinitely many virtual reads.
And while it may be feasible for reasonably small finite index
domains, large domains would require a potentially unfeasibly large number of
derivations to reach a saturated configuration.

Our extended calculus \calcext avoids the introduction of virtual reads
on constant arrays.
Instead, it propagates the \emph{default value} $v$ of a constant array~\aconst{v}
to array $a$
while keeping track of updated indices \indicesof{a}{\aconst{v}} between
$a$ and \aconst{v}.

\begin{definition}[Updated Indices]
  Given a propagation step $\propath{a}{\aconst{v}} \coloneq (r,c)$,
  the set of \emph{updated indices} \indicesof{a}{\aconst{v}} consists of all
  indices updated while propagating the default value of \aconst{v} to
  array~$a$, and is defined as follows.
  \begin{equation*}
    \indicesof{a}{\aconst{v}} =
    \begin{cases}
      \emptypair                           &\text{if } \propath{a}{\aconst{v}} = \emptypair\\
      \emptyset                            &\text{if } a = \aconst{v}\\
      \indicesof{b}{\aconst{v}} \cup \{j\} &\text{if }
      \propath{a}{\aconst{v}}=(\btrue,b) \text{ with } \\&
      \hspace{2.57em}b=\awrite{a}{j}{u} \text{ or }
      a=\awrite{b}{j}{u} \text{ for some }u \\
      \indicesof{c}{\aconst{v}}            &\text{otherwise, where } \propath{a}{\aconst{v}} = (r,c)\\
    \end{cases}
  \end{equation*}
  \qed
\end{definition}

Propagating the default value of \aconst{v} to array $a$
is denoted as \propath{a}{\aconst{v}},
and the reason for this propagation
is denoted as \reasonof{a}{\aconst{v}}.
If a default value $v$ is propagated to array $a$,
indices of $a$ that are not updated along the propagation path
store this default value.

\begin{definition}[Array Default Value]
  Given array $a$ and constant array \aconst{v},
  if $\propath{a}{\aconst{v}} \ne \emptypair$,
  the default value of array $a$ at indices $i \not\in
  \indicesof{a}{\aconst{v}}$ is $v$, i.e.,
  \bforall{i : \indexsort}{\imp{i \not\in \indicesof{a}{\aconst{v}}}{\teq{\aread{a}{i}}{v}}}.
  Otherwise, it is uninterpreted.
  \qed
\end{definition}

\Cref{fig:calcext} lists the rules
for propagating array default values by modifying \propathf
(rules \ruleinitc, \rn{Cow*}, \rn{CEq*}), and
conflict rule \rulecongc for extending \assertions with a new constant array
congruence lemma.

Rule \ruleinitc initializes map \propathf for each constant array
$\aconst{v}\in\terms{\assertions}$ with its default value.
Rules \rulecowd and \rulecowu propagate default values over array stores
\awrite{a}{j}{u}.
These two rules are only applied if there exists an index $i$
that has not been updated on the propagation path from~\aconst{v} to array
store \awrite{a}{j}{u} and is different from $j$.
For infinite index domains, there always exists such an index $i$.
Rules \ruleceqr and \ruleceql propagate default values over array equalities
$\teq{a}{c}\in\terms{\assertions}$ from left to right and vice versa.
Rule \rulecongc handles the case when two default values $v$ and $w$ with
$\I \models \tneq{v}{w}$ are propagated to the same indices of array $a$.

\begin{example}
  Consider formula $\varphi$ from \Cref{ex:finite}, defined over finite constant arrays
  with bit-vectors of size 2 as indices and Booleans as elements.
  Our initial configuration is given as follows.
  \begin{align*}
    \assertions &\coloneq
    \teq{
      \overbrace{\awrite{\aconst{v}}{i_1}{u_1}}^{s_1}}{
      \overbrace{\awrite{a}{j_1}{u_2}}^{s_2}
    }
    \wedge
    \teq{
      \overbrace{\awrite{a}{i_2}{u_3}}^{s_3}}{
      \overbrace{\awrite{\aconst{w}}{j_2}{u_4}}^{s_4}
    }
    \wedge
    \tneq{v}{w}\\
    \I_0 &\coloneq \none\\
    \propathf_0 &\coloneq \emptymap
  \end{align*}
  After applying rules \ruleinitw and \ruleinitc for all stores and constant
  arrays, we get
  \begin{align*}
    \propath{s_1}{\aread{s_1}{i_1}} \coloneq (\btrue,s_1),
    &&\propath{s_3}{\aread{s_3}{i_2}} \coloneq (\btrue,s_3),
    &&\propath{\aconst{v}}{\aconst{v}} \coloneq (\btrue,\aconst{v}),\\
    \propath{s_2}{\aread{s_2}{j_1}} \coloneq (\btrue,s_2),
    &&\propath{s_4}{\aread{s_4}{j_2}} \coloneq (\btrue,s_4),
    &&\propath{\aconst{w}}{\aconst{w}} \coloneq (\btrue,\aconst{w}).
  \end{align*}
  Now, assume rule \ruleinterp yields the following interpretation \I:
  \begin{align*}
    i_1 = j_1 = i_2 = j_2,
    u_1 = u_2 = u_3 = u_4,
    v \ne w,
    s_1 = s_2,
    s_3 = s_4,
    \\
    \aread{s_1}{i_1} = u_1,
    \aread{s_2}{j_1} = u_2,
    \aread{s_3}{i_2} = u_3,
    \aread{s_4}{j_2} = u_4.
  \end{align*}
  We have that
  $\propath{\aconst{v}}{\aconst{v}} \ne \emptypair$ and
  $\propath{s_1}{\aconst{v}} = \emptypair$ for
  $s_1 = \awrite{\aconst{v}}{i_1}{u_1} \in \terms{\assertions}$,
  and $\indicesof{\aconst{v}}{\aconst{v}} = \{\}$.
  Since there exists an index not in $\{\} \cup \{ i_1^\I \}$,
  we can propagate \aconst{v} to $s_1$ by applying rule \rulecowu, which
  yields
  $\propath{s_1}{\aconst{v}} \coloneq (\btrue,\aconst{v})$.
  Intuitively, after we propagated the default value of \aconst{v} to $s_1$, all indices of $s_1$,
  except the updated index~$i_1$, have default value $v$.
  Next, rule \ruleceqr can be applied and propagates~\aconst{v}
  from $s_1$ to $s_2$, which yields
  $\propath{s_2}{\aconst{v}} \coloneq (\teq{s_1}{s_2},s_1)$.
  Since
  $\propath{s_2}{\aconst{v}} \ne \emptypair$ and
  $\propath{a}{\aconst{v}} = \emptypair$ for
  $s_2 = \awrite{a}{j_1}{u_2} \in \terms{\assertions}$
  and there still exists an index not in $\{ i_1^\I \} \cup \{ j_1^\I \}$,
  we can propagate \aconst{v} down to array $a$ via rule \rulecowd, which
  yields
  $\propath{a}{\aconst{v}} \coloneq (\btrue,s_2)$.
  Note that under interpretation \I, indices $i_1$ and $j_1$ are equal and
  thus $\{ i_1^\I \} \cup \{ j_1^\I \}$ yields a set containing one index.
  Following the same reasoning steps,
  from $\propath{a}{\aconst{v}} \ne \emptypair$ we can propagate \aconst{v}
  via rule \rulecowu to $s_3$, afterwards apply rule \ruleceqr to propagate
  \aconst{v} to $s_4$, and finally propagate \aconst{v} down to \aconst{w},
  which yield the following configuration of \propathf:
  \begin{align*}
    \propath{s_1}{\aread{s_1}{i_1}} \coloneq (\btrue,s_1),
    &&\propath{s_3}{\aread{s_3}{i_2}} \coloneq (\btrue,s_3),
    &&\propath{\aconst{v}}{\aconst{v}} \coloneq (\btrue,\aconst{v}),\\
    \propath{s_2}{\aread{s_2}{j_1}} \coloneq (\btrue,s_2),
    &&\propath{s_4}{\aread{s_4}{j_2}} \coloneq (\btrue,s_4),
    &&\propath{\aconst{w}}{\aconst{w}} \coloneq (\btrue,\aconst{w}),\\
    \propath{s_1}{\aconst{v}} \coloneq (\btrue,\aconst{v}),
    &&\propath{s_2}{\aconst{v}} \coloneq (\teq{s_1}{s_2},s_1),
    &&\propath{a}{\aconst{v}} \coloneq (\btrue,s_2),\\
    \propath{s_3}{\aconst{v}} \coloneq (\btrue,a),
    &&\propath{s_4}{\aconst{v}} \coloneq (\teq{s_3}{s_4},s_3),
    &&\propath{\aconst{w}}{\aconst{v}} \coloneq (\btrue,s_4).
  \end{align*}
  We now have that
  $\propath{\aconst{w}}{\aconst{v}}\!\ne\!\emptypair$,
  $\propath{\aconst{w}}{\aconst{w}}\!\ne\!\emptypair$
  and $\indicesof{\aconst{w}}{\aconst{v}}\!=\!\{ i_1, j_1, i_2, j_2\}\!$,
  but indices $i_1$, $j_1$, $i_2$, $j_2$ are all equal under \I,
  and \tneq{v}{w}.
  This means that all indices $i \ne i_1^\I$ need to have the default
  $v$ and $w$ at the same time, which is not possible, and thus, rule
  \rulecongc adds the following lemma to \assertions.
  \begin{align*}
    \overbrace{
    \teq{s_1}{s_2}
    \wedge
    \teq{s_3}{s_4}
    }^{
      \reasonof{\aconst{w}}{\aconst{v}}
    }
    \;
    \wedge
    \overbrace{
    \btrue
    }^{
      \reasonof{\aconst{w}}{\aconst{w}}
    }
    \wedge
    \;\;
    \bexists{i}{\hspace{-1.8em}\bigwedge_{k \in \{i_1, j_1, i_2, j_2 \}}\hspace{-1.8em}\tneq{i}{k}}
    \implies
    \teq{v}{w}
  \end{align*}
  After that \I and \propathf are reset to $\I_0$ and $\propathf_0$, and
  we continue with a new interpretation~\I.
  In the next interpretation, indices $i_1$, $j_1$, $i_2$, $j_2$ will be
  forced to be different due to the added lemma, which may yield
  interpretation~\I:
  \begin{align*}
    i_1 \ne j_1 \ne i_2 \ne j_2,
    u_1 = u_2 = u_3 = u_4,
    u_2 \ne v,
    v \ne w,
    s_1 = s_2,
    s_3 = s_4,
    \\
    \aread{s_1}{i_1}= u_1,
    \aread{s_2}{j_1}= u_2,
    \aread{s_3}{i_2}= u_3,
    \aread{s_4}{j_2}= u_4.
  \end{align*}
  With this interpretation, after applying rules \ruleinitw and \ruleinitc,
  read \aread{s_2}{j_1} is propagated from $s_2$ to \aconst{v} by applying
  rules \rulerowu, \ruleeqr, and \rulerowd.
  This will result in a new lemma generated by rule \ruleroc, since
  $\propath{\aconst{v}}{\aread{s_2}{j_1}} \ne \emptypair$
  and $\I \models \tneq{\aread{s_2}{j_1}}{v}$ due to
  $\I \models \teq{\aread{s_2}{j_1}}{u_2} \wedge \tneq{u_2}{v}$.

  Reads \aread{s_1}{i_1}, \aread{s_3}{i_2}, and \aread{s_4}{j_2} will be
  propagated similarly and after all conflicts are resolved we can finally
  derive a saturated configuration, and thus, conclude with \emph{satisfiable}.

  The following interpretation \I that will lead to a saturated configuration, i.e.,
  by applying all rules until fixed point, will not yield any conflict.
  \begin{align*}
    \I \coloneq \{&
    i_1 \ne j_1 \ne i_2 \ne j_2,
    v = u_2 = u_4,
    w = u_1 = u_3,
    v \ne w,\\
    &s_1 = s_2,
    s_3 = s_4,
    \aread{s_1}{i_1} = u_1,
    \aread{s_2}{j_1} = u_2,
    \aread{s_3}{i_2} = u_3,
    \aread{s_4}{j_2} = u_4
    \}.
  \end{align*}
  \qed
\end{example}


\subsection{Correctness}

In this section, we establish the theorems for refutational and satisfiability
soundness of our base calculus presented in \Cref{sec:base}, and its extension
presented in \Cref{sec:ext}.
Refutational soundness ensures that the calculus never reports \unsat unless
the initial set of constraints is unsatisfiable.  Satisfiability soundness is
the complementary property: if the calculus \emph{does not} report \unsat,
then the initial set of constraints is satisfiable. This is often
referred to as completeness, but we find the term satisfiability soundness more
precise in this context.
We prove the refutational and satisfiability soundness of our calculus
\calcext \begin{arxiv}in the appendix\end{arxiv}\begin{cav}in an extended version of this paper~\cite{arxiv}\end{cav}.

\begin{theorem}[Refutational Soundness]
  If a derivation starting with an initial configuration \config{\assertions}{\I}{\propathf}
  ends in \unsat, then $\assertions$ is $\theoryarr$-unsatisfiable.
\end{theorem}

Refutational soundness holds for arbitrary derivations.
For satisfiability soundness, we need a restriction on the order in
which propagations of default values are recorded.
Since rules \rn{Cow*} and \rn{CEq*} record at most one propagation step per
pair $(a, \aconst{v})$, the recorded set of updated indices
\indicesof{a}{\aconst{v}} is determined by the \emph{first} recorded
propagation path.
This may result in an over-approximation of the indices updated along other
derivable propagation paths.
For finite index domains, this can lead to a conflict not detected by
\rulecongc if the existential side condition fails on the union of the
over-approximated index sets.

We therefore restrict derivations to a \emph{strategy} that records
\emph{conflicting propagations eagerly}.
That is, whenever rules \ruleinitc, \rn{Cow*}, and \rn{CEq*} admit sequences of
applications under \I that make rule \rulecongc applicable
for some $a$, \aconst{v} and \aconst{w},
we record the propagation steps that lead to that conflict
before any other propagation step.
This strategy is effective since only finitely many propagation paths,
that visit each each array at most once, need to be considered. Our
implementation realizes this strategy by intersecting the updated index sets
over all propagation paths explored under~\I.

\begin{theorem}[Satisfiability Soundness]
  If a derivation following the above strategy, starting with an initial
  configuration \config{\assertions}{\I_0}{\propathf_0},
  ends in a configuration other than \unsat, and every applicable rule
  is redundant, then $\assertions$ is $\theoryarr$-satisfiable.
\end{theorem}

\section{Integration}

In \Cref{sec:calc},
we presented our new calculus \calcext for the theory of extensional constant
arrays \theorycarr as a set of derivation rules that modify configurations
starting from the initial configuration \config{\assertions}{\I_0}{\propathf_0}.
Our calculus is \emph{independent} of the underlying solver architecture
and only requires the existence of a procedure that provides an interpretation
\I in the empty theory \theoryeuf (rule \ruleinterp).

We integrated a decision procedure for \theorycarr based on our new calculus
\calcext in the state-of-the-art SMT solver Bitwuzla~\cite{bitwuzla}.
Bitwuzla specializes on
theories relevant for bit-precise reasoning and
supports the theories of fixed-size bit-vectors, floating-point
arithmetic, and combinations with arrays and uninterpreted functions.
Until now, it implemented a decision procedure for the theory of extensional
arrays based on~\cite{BrummayerB09}, extended with support for nested arrays
and non-extensional constant arrays.
We replaced this procedure with a new decision procedure based on \calcext,
which seamlessly integrates into the LOD architecture of Bitwuzla.
In the context of \calcext, the bit-vector solver at the core of its LOD
architecture is responsible for providing interpretation \I~\cite{NiemetzPZ24}.
As implemented in Bitwuzla, our procedure incorporates several \emph{optimizations}.

\paragraph{\distinctn.}
When a conflict between two different constant arrays is detected, i.e.,
rule \rulecongc applies, we produce a resulting lemma that does not actually
contain a quantified formula.
Instead, the quantified formula is encoded as a special variant of \distinct, which we call
\distinctn.
While $\distinct(t_1, \ldots, t_k)$ is true iff terms
$t_1, \ldots, t_k$ are pairwise different,
$\distinctn(n, t_1, \ldots, t_k)$ is true iff at least $n$ terms
are pairwise different, i.e.,
the set of values assigned to $t_1, \ldots, t_k$ has at least $n$ elements.
If $n = k$,
then $\distinctn(k, t_1, \ldots, t_k) \equiv \distinct(t_1, \ldots, t_k)$,
while for $n > k$,
$\distinctn(n, t_1, \ldots, t_k) \equiv \bfalse$.
Given a finite sort \elementsort for $t_i$ and
$k > |\elementsort^\I|$, $\distinct(t_1, \ldots, t_k) \equiv \bfalse$,
whereas $\distinctn(n, t_1, \ldots, t_k)$ for $n \le k$ is not.
We encode the quantified formula in \rulecongc as follows.
\begin{align*}
  \bexists{i : \indexsort}{
    \hspace{-3.3em}\bigwedge_{
      k \in \indicesof{a}{\aconst{v}}\cup\indicesof{a}{\aconst{w}}}\hspace{-3.2em}\tneq{i}{k}}
  \;\;\equiv\;\;
  \neg\distinctn(|\indexsort^\I|,\indicesof{a}{\aconst{v}}\cup\indicesof{a}{\aconst{w}})
\end{align*}
The intuition here is that this quantified formula is true if there exists an
index~$i$ that has not been updated between \aconst{v} and $a$ and \aconst{w}
and $a$, and is thus not included in
$\I_{\cup} \coloneq \indicesof{a}{\aconst{v}}\cup\indicesof{a}{\aconst{w}}$.
For \emph{infinite} index domains, this will always be \btrue.
However, for \emph{finite} index domains it depends on whether the set of updated
indices $\I_{\cup}$
includes all possible indices of index domain $\indexsort^\I$, which can be
encoded as
$\distinctn(|\indexsort^\I|,\indicesof{a}{\aconst{v}}\cup\indicesof{a}{\aconst{w}})$.

As an additional optimization, we encode a given
$\distinctn(n, t_1, \ldots, t_k)$ constraint lazily and use a
custom decision heuristic that tries to assign different values to terms
$t_1, \ldots, t_k$.
If our heuristic fails to assign $n$ unique values to the given terms, we
lazily expand the encoding of the constraint until we
find $n$ different values or conclude that $\distinctn(n, t_1, \ldots, t_k) \equiv \bfalse$.

\paragraph{Propagation.}
In order to reduce memory overhead, we do not immediately record pair $(r, c)$
when constructing \propath{a}{t} while exploring all possible propagations of reads.
Since these pairs are only required for constructing \reasonoff,
we instead replay the propagation steps
to produce the full map \propathf containing $(r,c)$
only for the array reads involved in a conflict.
In this replay phase, propagation is
performed in breadth-first-search order, which yields the shortest sequence
of propagation steps to the conflict under the current interpretation.

\section{Evaluation}

We evaluate the performance of our procedure for the theory of extensional
constant arrays as implemented in the SMT solver Bitwuzla~\cite{bitwuzla}, denoted as
configuration \configbzlaconst.
Bitwuzla supports the (quantified and quantifier-free)
theories of fixed-size bit-vectors and floating-point arithmetic, in
combination with arrays and uninterpreted functions.
We therefore focus on problems over \emph{finite} constant arrays.

We compare our configuration \configbzlaconst against its baseline version
\configbzlabase (commit 961ccbc5~\cite{bitwuzla-gh}), and the SMT solvers
\configcvcv~\cite{cvc5} (version 1.3.2), \configziii~\cite{z3}
(commit 9666915d~\cite{z3-gh}) and \configmsat~\cite{mathsat5} (version 5.6.16).
These four support different \emph{fragments} of theory \theorycarr.
In particular, our configuration \configbzlaconst is the only SMT solver that
\emph{fully} supports \theorycarr for finite constant arrays.
Note that \configziii and \configmsat allow the use of finite constant arrays
without restrictions,
but we encountered soundness issues that suggest that they only support the
fragment of \theorycarr \emph{without} extensionality over finite constant arrays.
We exclude Boolector~\cite{btor2} and SMTInterpol\cite{smtinterpol} from
the evaluation, since they do not support extensionality over finite constant
arrays.
We further exclude \configmsat from quantified benchmark sets
since it does not support quantifier reasoning.

Support for constant arrays is not yet standardized in the SMT-LIB
language. Thus, the SMT-LIB benchmark
library~\cite{smtlib25inc,smtlib25noninc} does not include
benchmarks with constant arrays.
We therefore evaluated our approach in different usage scenarios: a) on
benchmarks without constant arrays but where reasoning about constant arrays
is required internally to Bitwuzla;
and b) on a compilation of diverse sets of benchmarks with constant arrays, one
crafted and the others stemming from different applications.
We use the following benchmark sets.

\paragraph{\emph{\benchsmtlib (5,112 benchmarks).}}
The non-incremental benchmarks of the quantified logics \texttt{ABV},
\texttt{ABVFP} and \texttt{ABVFPLRA} of the SMT-LIB benchmark
library~\cite{smtlib25noninc}.
These benchmarks do not contain constant arrays but require constant
array reasoning in subsolver queries of the model-based quantifier
instantiation (MBQI)~\cite{GeM09} engine implemented in Bitwuzla.
This is due to the MBQI procedure asserting the model value of each
constant $x$ that occurs in the ground formulas via equalities of the form
\teq{x}{x^\I}.
The model value of an array constant is represented as a store
chain on a constant-initialized array.
This inevitably introduces equalities over constant arrays.
We use this set to evaluate how much our new procedure improves quantified
reasoning with arrays in Bitwuzla.

\paragraph{\emph{\benchmbqi (4,147 benchmarks).}}
A set of quantifier-free non-incremental benchmarks defined over
value-initialized and nested extensional constant arrays with fixed-size
bit-vectors.
These benchmarks were extracted from subsolver queries of the MBQI solver of
\configbzlaconst on benchmarks from set \emph{smtlib}
that were unsolved by baseline configuration \configbzlabase due to subsolver
queries with extensional constant arrays.
We included at most one satisfiable and one unsatisfiable of these queries
per corresponding \benchsmtlib benchmark.

\paragraph{\emph{\benchpono (41 benchmarks, 2,838 incremental queries).}}
A set of quantifier-free incremental benchmarks over zero-initialized
extensional constant arrays with fixed-size bit-vectors.
This set was extracted from all queries with constant arrays generated by the
k-induction engine of the model checker Pono~\cite{pono} on benchmarks of the
hardware model checking competitions (HWMCC)
2019--2025~\cite{hwmcc19,hwmcc20,hwmcc24,hwmcc25}.
Similarly to Bitwuzla's MBQI engine, this model-checking engine produces queries
that equate array constants with array model values.

\paragraph{\emph{\benchhevm (34 benchmarks).}}
This set contains benchmarks generated by
hevm~\cite{hevm,hevmbench},
a symbolic execution engine for the Ethereum virtual machine.
Benchmarks in this set are encoded over bit-vectors of
size~256 in combination with uninterpreted functions and zero-initialized
extensional constant arrays.

\paragraph{\emph{\benchcraft(410 benchmarks).}}
In order to cover a wider range of the encoding capabilities (and corner cases) of
\theorycarr with finite constant arrays, we created a set of
benchmarks as the instantiations of parameterizing formula
\begin{align*}
{\awrite{\awrite{\awrite{\aconst{v}}{i_1}{e_1^i}}{\ldots}{\ldots}}{i_n}{e_n^i}}
\teqf
{\awrite{\awrite{\awrite{a_1}{k_1}{e_1^k}}{\ldots}{\ldots}}{k_s}{e_s^k}}
\teqf\ldots\\
\ldots\teqf
{\awrite{\awrite{\awrite{a_z}{l_1}{e_1^l}}{\ldots}{\ldots}}{l_t}{e_t^l}}
\teqf
{\awrite{\awrite{\awrite{\aconst{w}}{j_1}{e_1^j}}{\ldots}{\ldots}}{j_m}{e_m^j}}
\end{align*}
over the number of array constants $z$ and
the number of array updates $n$, $m$, $s$ and~$t$.
Indices $\{i_1,\ldots,i_n, k_1,\ldots,k_s, l_1,\ldots,l_t, j_1,\ldots,j_m\}$
and elements $\{e_1^i,\ldots,\allowbreak e_n^i, \allowbreak e_1^k,\ldots,e_s^k,\allowbreak
e_1^l,\ldots,e_t^l,e_1^j,\ldots,e_m^j\}$ are uninterpreted
constants over finite sorts
(\bool, bit-vector, floating-point arithmetic, rounding mode).
This benchmark set is split into the two subsets \benchcraftc and \benchcraftq
($205$ benchmarks each).
Set \benchcraftc encodes constant arrays natively
via \texttt{(as const \ldots)}, and set \benchcraftq encodes
constant arrays with quantifiers via axiom~\ref{eq:roc}.
The quantified encoding serves two purposes:
(1) as a baseline to compare against the native encoding, i.e.,
to show whether the native encoding has any advantages over an axiomatization
with quantifiers; and
(2) if solvers disagree on the native encoding, we can compare their answers
to the results on the quantified encoding to determine
which solver gave a wrong answer.

\begin{table}[t]
  \caption{Number of solved/satisfiable/unsatisfiable/unsolved benchmarks/queries
  (Solved,Sat,Unsat,Unsolved),
  and penalized runtime (T), grouped by benchmark set.
  Number ($x$/$y$) indicates the total number of benchmarks $x$ and
  satisfiability queries $y$ in the set.
  Numbers marked with \textsuperscript{*} include incorrect answers.
  }
  \label{tab:eval:overall}
  \resizebox{\textwidth}{!}{
  \input{table/results.tex}
  }
\end{table}

For each solver and benchmark pair, we allocated one CPU core
and 8GB of memory with a time limit of 1200 seconds.
If a solver terminated with \emph{unknown}, an error or ran into the
memory limit on a specific benchmark, we counted its runtime on that benchmark
as 1200 seconds as a penalty.

\Cref{tab:eval:overall} summarizes the results of each solver over all
benchmark sets.
Overall, \configbzlaconst significantly outperforms \emph{all other solvers} on each
benchmark set.

On the \benchsmtlib set, our new procedure significantly improves the
performance of Bitwuzla's MBQI engine on problems where MBQI subqueries involve
reasoning over extensional constant arrays.
Compared to the baseline version \configbzlabase, the number of solved
benchmarks increases by a factor of 5 (3,545 vs.~717).
On the 636 benchmarks commonly solved by \configbzlaconst and \configbzlabase,
we observe a $27\times$ speed-up in runtime (22 vs.~594.5 seconds).

On the \benchmbqi set, \configbzlaconst solves all queries, including the ones that
\configbzlabase was not able to solve due to extensional constant arrays. This
observed performance gain directly translates to the improved performance on
the \benchsmtlib set, since \configbzlaconst is now able to efficiently solve the
MBQI subqueries.
On the 4,126 instances commonly solved by \configziii and \configbzlaconst,
\configbzlaconst achieves a $140\times$ speed-up in solving time over \configziii
(37.8 vs.~5,279.6 seconds).
Compared to \configmsat, \configbzlaconst achieves a $1.4\times$ speed-up
on the 4,034 commonly solved instances (44.9 vs.~63.9 seconds).

On benchmark set \benchpono,
we observe a similar performance improvement of
\configbzlaconst over \configbzlabase, where \configbzlaconst is able to solve
all incremental queries on which \configbzlabase struggled due to
extensional constant arrays.

On the \benchhevm set,
both \configbzlabase
and \configbzlaconst solve all benchmarks the fastest.
Note that these benchmarks only make use of non-extensional zero-initialized
constant arrays.
Thus, \configbzlabase and \configbzlaconst are expected to perform similarly.
The difference in runtime is due to one benchmark where Bitwuzla's
abstraction-refinement approach for bit-vector arithmetic~\cite{NiemetzPZ24} receives different bit-vector
models, which helps \configbzlaconst to converge faster.

On the \benchcraftc set,
\configbzlaconst solves the highest number of benchmarks.
This benchmark set consists of benchmarks that encode equalities over finite
constant arrays with different default values.
Both \configcvcv and baseline \configbzlabase do not support this case and reject the
input with an appropriate error message.
\configziii and \configmsat, on the other hand, accept the input, but report
wrong answers for a significant number of benchmarks.
In contrast, we did not encounter any disagreements between
any of the solvers on the \benchcraftq set.

We investigated all of the disagreements on the native encoding and compared
the solver results with the results of all solvers on the quantified encoding
in \benchcraftq.
We found that \configziii reports \emph{sat} (resp.~\emph{unsat}) on 31 (19) benchmarks that
are known to be \emph{unsatisfiable} (resp.~\emph{satisfiable}) by
construction.
Further, on benchmarks commonly solved by \configziii between \benchcraftc and \benchcraftq,
it disagrees with itself on 1 benchmark with unknown status
and 5 benchmarks with known status.
\configmsat reports \emph{unsat} on 29 benchmarks with unknown status
that are determined to be satisfiable by
\configbzlaconst and \configbzlabase on the quantified encoding.
Further, 17 of these benchmarks are also solved by \configbzlaconst on the
native encoding and determined to be satisfiable.

In an additional experiment,
we investigated the models produced by
\configbzlaconst, \configmsat, and \configziii on the \benchcraftc set
for correctness.
That is, when a solver reported \emph{sat}, we asked the solver to produce
a model, which was then asserted back to the original formula.
We ran all three solvers on all resulting benchmarks
and found that the \emph{only} solver that did \emph{not} disagree with its own
models was \configbzlaconst.
Further, \configbzlaconst rejected 3 (out of 12) \configmsat models and
69 (out of 84) \configziii models.
These findings suggest that both \configmsat and \configziii, while accepting
the input, do \emph{not} fully support equality over constant arrays with
finite index domains.
We contacted the developers of \configmsat and \configziii and provided the
failing test cases.

Finally,
with the native encoding of constant arrays,
\configbzlaconst is able to solve significantly more benchmarks than with the
quantifier-based encoding \benchcraftq.
Overall, only 13 benchmarks that could be solved by at least one solver with the
quantifier-based encoding could not be solved by \configbzlaconst
with the native encoding.
In contrast, \configbzlaconst was able to solve 59 native benchmarks that
none of the solvers was able to solve with the quantified encoding.
On closer inspection of the behavior of \configbzlaconst on native benchmarks
that were hard to solve, we found that the majority of the solving time is
spent on finding an index arrangement via \distinctn.
We leave improving the lazy expansion scheme for \distinctn to future work.

\section{Conclusion}
We have presented a new calculus \calcext for the theory of extensional
constant arrays \theorycarr without imposing any limitations on the finiteness
of the array index domain.
We have shown the refutational and satisfiability soundness of \calcext.
We integrated a new decision procedure for \theorycarr based on \calcext
in the state-of-the-art
SMT solver Bitwuzla and evaluated its performance on a diverse set of
benchmarks.
The new decision procedure enables Bitwuzla to solve problems it was not able
to handle before, and thus significantly improves over its baseline.

While the integration into Bitwuzla and its evaluation has focused
on arrays with finite index domains, our procedure is agnostic to the
underlying solver architecture and fully supports the infinite domain case.
As future work, we plan to integrate a new array solver based on \calcext
in the SMT solver cvc5.

\begin{credits}
\subsubsection{\ackname} \papthanks

\subsubsection{\discintname}
The authors have no competing interests to declare that are
relevant to the content of this article.
\end{credits}

\subsubsection*{Data Availability Statement.}
The artifact accompanying this paper is ar\-chived and available in the Zenodo
repository at \url{https://zenodo.org/records/19713260}.


\bibliographystyle{splncs04}
\bibliography{references}

\begin{arxiv}
\clearpage
\appendix
\section*{Appendix}

In this appendix, we prove refutational and satisfiability soundness of our base
calculus presented in \Cref{sec:base}, and its extension presented in
\Cref{sec:ext}.

\vspace{2ex}
We start with several crucial lemmas that are needed to prove the main
theorems.
First, we need the following definition, which assigns a
specific numeric value indicating the recursion depth necessary to calculate a
reason $\reasonoff$.
\begin{definition}
  \begin{equation*}
  |\reasonof{a}{t}| =
  \begin{cases}
    0                     &\text{ if } \propath{a}{t} = \emptypair\\
    1                     &\text{ if } t = a \text{ or } t = a[i] \text{ for some } i\\
    1 + |\reasonof{c}{t}| &\text{ otherwise, where } \propath{a}{t}=(r,c)
  \end{cases}
  \end{equation*}
  \qed
\end{definition}

The first lemma states that once a reason $\reasonoff$ has been
established, it does not change unless a conflicting rule is applied and the
configuration is reset.
\begin{lemma}\label{lem:reasonprops}
  Let $D$ be a derivation ending in configuration $C$, with $C\not=\unsat$.  The following are true
  of $\reasonoff$ in $C$ for all $a$ and $t$:
  ($i$) if $C$ is obtained by
  applying a non-conflict rule to a configuration $C'$, and
  $\reasonof{a}{t}\not=\emptypair$ in $C'$,
  then both $\reasonof{a}{t}$ and $|\reasonof{a}{t}|$ remain unchanged
  from $C'$ to $C$.
  ($ii$) $\reasonof{a}{t}$ is well defined;
  ($iii$) $|\reasonof{a}{t}|$ is finite; and
  ($iv$) $\reasonof{a}{t}=\emptypair$ iff
  $\propath{a}{t}=\emptypair$.
\end{lemma}
\begin{proof}
  The proof is by induction over deductions.  In an initial configuration ($i$)
  holds trivially, and $\reasonoff$ and $\propathf$ are always equal to $\emptypair$,
  so it is easy to see that the properties hold.

  Now, suppose they hold for configuration $C'$, and let $C$ be obtained from $C'$ by
  the application of a rule other than \ruleconf.  Suppose the rule is a
  conflict rule.  Then ($i$) holds trivially, and as above, $\reasonoff$ and
  $\propathf$ are always equal to $\emptypair$, so the other three properties hold as well.

  Suppose the rule is not a conflict rule.
  We first show ($i$) by considering, in $C'$, the cases in the definition of
  $\reasonoff$.  Suppose 
  $\propath{a}{t}=\emptypair$ in $C'$.  By ($iv$), we also have
  $\reasonof{a}{t}=\emptypair$ in $C'$, so the property holds trivially.
  Suppose $t=a$ or $t=a[i]$ for some $i$, then $\reasonof{a}{t}=\btrue$ and
  $|\reasonof{a}{t}|=1$ regardless of the
  configuration.  Finally, assume we have (in $C'$)
  $\reasonof{a}{t}=\reasonof{c}{t} \wedge r \text{, where } (r,c) = \propath{a}{t}$.
  And suppose ($i$) does not hold.
  Then we can choose $a$ and $t$ such that
  $\reasonof{a}{t}\not=\emptypair$ in $C'$,
  $\reasonof{a}{t}$ or $|\reasonof{a}{t}|$ changes value from $C'$ to $C$, and
  $|\reasonof{a}{t}|$ is minimum in $C'$ (possible by property ($iii$)).
  
  But $\propath{a}{t}$ cannot have changed from $C'$ to $C$, since
  all non-conflict rules only modify $\propath{a}{t}$ if its value
  is $\emptypair$. So, we must also have in $C$ that
  $\reasonof{a}{t}=\reasonof{c}{t} \wedge r \text{, where } (r,c) =
  \propath{a}{t}$ and
  $|\reasonof{a}{t}|=|\reasonof{c}{t}| + 1$.
  Thus, if $\reasonof{a}{t}$ or $|\reasonof{a}{t}|$ has changed, it must be
  because either $\reasonof{c}{t}$ or $|\reasonof{c}{t}|$ has also changed.  But by definition,
  $|\reasonof{c}{t}| = |\reasonof{a}{t}| - 1$, which
  contradicts the minimality assumption.

  We now show ($ii$) and ($iii$).  Suppose $\reasonof{a}{t}\not=\emptypair$ in
  $C'$.  We know by the induction hypothesis that $\reasonof{a}{t}$ is well
  defined and $|\reasonof{a}{t}|$ is finite in $C'$.
  As we just showed above, $\reasonof{a}{t}$ and $|\reasonof{a}{t}|$
  remain unchanged in $C$, meaning that
  $\reasonof{a}{t}$ is well defined and $|\reasonof{a}{t}|$ is finite in $C$.
  On the other hand, suppose $\reasonof{a}{t}=\emptypair$ in $C'$.  By
  $(iv)$, we have $\propath{a}{t}=\emptypair$ in $C'$ as well. If
  $\propath{a}{t}=\emptypair$ still in $C$, then it is clear that
  $\reasonof{a}{t}=\emptypair$ in $C$ as well, and so it is well
  defined, and $|\reasonof{a}{t}|=0$.  The last case is when
  $\propath{a}{t}=\emptypair$ in $C'$ but
  $\propath{a}{t}\not=\emptypair$ in $C$.  This happens when the
  rule newly defines $\propath{a}{t}$.  But whenever this happens,
  we either have $t=a$ or $t=a[i]$ for some $i$, in which case
  $\reasonof{a}{t}$ is trivially well defined with $|\reasonof{a}{t}|=1$, or
  the rule newly defines 
  $\propath{a}{t}$ to be some $(r,a')$, where $\propath{a'}{t}\not=\emptypair$
  in $C'$.  In the latter case, by ($iv$), $\reasonof{a'}{t}\not=\emptypair$ in $C'$, so by
  ($i$), $\reasonof{a'}{t}$ and $|\reasonof{a'}{t}|$ are unchanged from $C'$ to
  $C$. And by ($iii$), $|\reasonof{a'}{t}|$ is finite in $C'$. But, by
  definition, $\reasonof{a}{t}=\reasonof{a'}{t} \wedge r$, which is
  well defined, and similarly, $|\reasonof{a}{t}|=|\reasonof{a'}{t}|+1$, which
  is finite.

  Finally, we show ($iv$).  The if direction follows directly by the
  definition of $\reasonof{a}{t}$.  For the only-if direction, consider the
  contrapositive.  Suppose that $\propath{a}{t}\not=\emptypair$.  Then by
  definition, either $\reasonof{a}{t}=\btrue$ or $\reasonof{a}{t}=r_1 \wedge
  r_2$ for some $r_1,r_2$.  We know that it is well defined by ($ii$), so in
  either case, $\reasonof{a}{t}\not=\emptypair$.
  \qed
\end{proof}

The next lemma provides the main invariant needed for soundness in the base calculus \calcbase.

\begin{lemma}\label{lem:invariant}
  Let $D$ be a derivation ending in configuration $C$, with $C\not=\unsat$.  The following is true
  in $C$ for all $a,b,i$:
\[\text{If }\propath{a}{\aread{b}{i}}\not=\emptypair\text{, then }
\theoryarr \models \reasonof{a}{\aread{b}{i}} \implies \teq{\aread{a}{i}}{\aread{b}{i}}.\]
\end{lemma}
\begin{proof}
  The proof is by induction on the length of a derivation. If $C$ is the
  initial configuration, $\propathf$ is
  empty, so the statement holds vacuously.

  Now, suppose $C$ is obtained by applying a rule to configuration $C'$. We
  consider each rule and show that it preserves the property.

  First, if the rule is a conflict rule, then $\propathf$ is empty in $C$, so
  again, the statement holds vacuously.
  
  Second, if $\propath{a}{\aread{b}{i}}\not=\emptypair$ in $C'$, then
  by~\Cref{lem:reasonprops}, $\reasonof{a}{\aread{b}{i}}$ has the same value in $C$ and
  $C'$, so, since the property holds in $C'$ by the induction hypothesis, the
  property must hold also in $C$.

  The remaining case to consider is when $\propath{a}{\aread{b}{i}}$ is newly defined
  by the rule.  We consider all such rules.
  First note that the \ruleinitr and \ruleinitw rules define $\propathf$ for
  a pair of arguments of the form $(a,a[i])$.  Since $\teq{a[i]}{a[i]}$ is
  always true, the property clearly holds in $C$.

  Next, consider \rulerowd.  The newly defined value of $\propath{a}{\aread{b}{i}}$ is
  $(\tneq{i}{j},\awrite{a}{j}{u})$.  We also know that
  $\propath{\awrite{a}{j}{u}}{\aread{b}{i}} \ne \emptypair$ in $C'$,
  so we also have $\theoryarr \models \reasonof{\awrite{a}{j}{u}}{\aread{b}{i}}
  \implies \teq{\aread{\awrite{a}{j}{u}}{i}}{\aread{b}{i}}$ in $C'$ by the
  induction hypothesis. The same holds also in $C$ by~\Cref{lem:reasonprops}. Now, by the definition of $\reasonoff$, we have (in $C$)
  $\reasonof{a}{\aread{b}{i}}=\reasonof{\awrite{a}{j}{u}}{\aread{b}{i}}\wedge\tneq{i}{j}$.
  And by the~\eqref{eq:rowne} array axiom, we have $\theoryarr \models
  \tneq{i}{j} \implies \teq{\aread{a}{i}}{\aread{\awrite{a}{j}{u}}{i}}$.
  It follows that (in $C$)
  $\theoryarr\models\reasonof{a}{\aread{b}{i}}\implies\teq{\aread{a}{i}}{\aread{b}{i}}$.

  The \rulerowu case is similar.  Suppose $a=\awrite{a'}{j}{u}$.  The newly defined value of $\propath{\awrite{a'}{j}{u}}{\aread{b}{i}}$ is
  $(\tneq{i}{j},a')$.  We also know that
  $\propath{a'}{\aread{b}{i}} \ne \emptypair$ in $C'$,
  so we also have $\theoryarr \models \reasonof{a'}{\aread{b}{i}}
  \implies \teq{\aread{a'}{i}}{\aread{b}{i}}$ in $C'$ by the
  induction hypothesis. The same holds also in $C$
  by~\Cref{lem:reasonprops}. Now, by the definition of $\reasonoff$, we have
  (in $C$) $\reasonof{\awrite{a'}{j}{u}}{\aread{b}{i}}=\reasonof{a'}{\aread{b}{i}}\wedge\tneq{i}{j}$.
  And by the~\eqref{eq:rowne} array axiom, we have $\theoryarr \models \tneq{i}{j} \implies
  \teq{\aread{a'}{i}}{\aread{\awrite{a'}{j}{u}}{i}}$.  It follows that (in $C$)
  $\theoryarr\models\reasonof{\awrite{a'}{j}{u}}{\aread{b}{i}}\implies\teq{\aread{\awrite{a'}{j}{u}}{i}}{\aread{b}{i}}$.

  Consider now \ruleeql.  The newly defined value of
  $\propath{a}{\aread{b}{i}}$ in $C$ is
  $(\teq{a}{c},c)$.  We also know that
  $\propath{c}{\aread{b}{i}} \ne \emptypair$ in $C'$,
  so we also have $\theoryarr \models \reasonof{c}{\aread{b}{i}}
  \implies \teq{\aread{c}{i}}{\aread{b}{i}}$ in $C'$ by the
  induction hypothesis. The same holds also in $C$ by~\Cref{lem:reasonprops}. Now, by the definition of $\reasonoff$, we have (in $C$)
  $\reasonof{a}{\aread{b}{i}}=\reasonof{c}{\aread{b}{i}}\wedge\teq{a}{c}$.
  And by equality reasoning, we have $\theoryarr \models \teq{a}{c} \implies
  \teq{\aread{a}{i}}{\aread{c}{i}}$.  It follows that
  $\theoryarr\models\reasonof{a}{\aread{b}{i}}\implies\teq{\aread{a}{i}}{\aread{b}{i}}$.
  The \ruleeqr case is analogous.

  These are the only rules that add a newly defined value to $\propathf$ of the
  form $(a,\aread{b}{i})$.
  \qed
  \end{proof}

  Next, for our extended calculus \calcext we need an additional definition
  and a pair of lemmas similar to the \calcbase case.

\begin{definition}
  \begin{equation*}
  |\indicesof{a}{\aconst{v}}| =
  \begin{cases}
    0                     &\text{if } \propath{a}{\aconst{v}} = \emptypair\\
    1                     &\text{if } a = \aconst{v}\\
    1 + |\indicesof{b}{\aconst{v}}| &\text{if } 
      \propath{a}{\aconst{v}}=(\btrue,b) \text{ with } \\&
      \hspace{2.57em}b=\awrite{a}{j}{u} \text{ or }
      a=\awrite{b}{j}{u} \text{ for some }u \\
    1 + |\indicesof{c}{\aconst{v}}| &\text{otherwise, where } \propath{a}{\aconst{v}}=(r,c)
  \end{cases}
  \end{equation*}
  \qed
\end{definition}

\begin{lemma}\label{lem:indicesprops}
  Let $D$ be a derivation ending in configuration $C$, with $C\not=\unsat$.
  The following are true
  of $\indices$ in $C$ for all $a$ and $t$:
  ($i$) if $C$ is obtained by
  applying a non-conflict rule to a configuration $C'$, and
  $\indicesof{a}{\aconst{v}}\not=\emptypair$ in $C'$,
  then both $\indicesof{a}{\aconst{v}}$ and $|\indicesof{a}{\aconst{v}}|$
  remain unchanged from $C'$ to $C$.
  ($ii$) $\indicesof{a}{\aconst{v}}$ is well defined; ($iii$)
  $|\indicesof{a}{\aconst{v}}|$ is finite; and
  ($iv$) $\indicesof{a}{\aconst{v}}=\emptypair$ iff
  $\propath{a}{\aconst{v}}=\emptypair$.
\end{lemma}
\begin{proof}
  The proof is by induction over deductions.  In an initial configuration ($i$)
  holds trivially, and $\indices$ and $\propathf$ are always equal to $\emptypair$,
  so it is easy to see that the properties hold.

  Now, suppose they hold for configuration $C'$, and let $C$ be obtained from $C'$ by
  the application of a rule other than \ruleconf.  Suppose the rule is a
  conflict rule.  Then ($i$) holds trivially, and as above, $\indices$ and
  $\propathf$ are always equal to
  $\emptypair$, so the other three properties hold as well.

  Suppose the rule is not a conflict rule.
  We first show ($i$) by considering, in $C'$, the cases in the definition of
  $\indices$.  Suppose
  $\propath{a}{t}=\emptypair$ in $C'$.  By ($iv$), we also have
  $\indicesof{a}{\aconst{v}}=\emptypair$ in $C'$, so the property holds trivially.
  Suppose $a=\aconst{v}$; then, $\indicesof{a}{\aconst{v}}=\emptyset$ and
  $|\indicesof{a}{\aconst{v}}|=1$ regardless of the
  configuration.
  Finally, assume we have (in $C'$)
  $\indicesof{a}{\aconst{v}}=\indicesof{b}{\aconst{v}}\cup\{j\}$ (for some
  $b,j$) or $\indicesof{a}{\aconst{v}}=\indicesof{c}{\aconst{v}}$ (for some
  $c$), and suppose ($i$) does not hold.
  Then we can choose $a$ and $v$ such that
  $\indicesof{a}{\aconst{v}}\not=\emptypair$ in $C'$,
  $\indicesof{a}{\aconst{v}}$ or $|\indicesof{a}{\aconst{v}}|$ changes value from $C'$ to $C$, and
  $|\indicesof{a}{\aconst{v}}|$ is minimum in $C'$ (possible by ($iii$)). 
  
  But $\propath{a}{\aconst{v}}$ cannot have changed from $C'$ to $C$, since
  all non-conflict rules only modify $\propath{a}{\aconst{v}}$ if its value
  is $\emptypair$.  So, if $\indicesof{a}{\aconst{v}}$ or
  $|\indicesof{a}{\aconst{v}}|$ has changed, it must be
  because of a change in $\indicesof{b}{\aconst{v}}$ or $|\indicesof{b}{\aconst{v}}|$ (in the first case) or
  because of a change in $\indicesof{c}{\aconst{v}}$ or $|\indicesof{c}{\aconst{v}}|$(in the second case).
  But by definition, we have
  $|\indicesof{b}{\aconst{v}}| = |\indicesof{a}{\aconst{v}}| - 1$ or
  $|\indicesof{c}{\aconst{v}}| = |\indicesof{a}{\aconst{v}}| - 1$, which in
  either case contradicts the minimality assumption.

  We now show ($ii$) and ($iii$).  Suppose $\indicesof{a}{\aconst{v}}\not=\emptypair$ in
  $C'$.  We know by the induction hypothesis that $\indicesof{a}{\aconst{v}}$
  is well defined and $|\indicesof{a}{\aconst{v}}|$ is finite in $C'$.
  As we just showed above, $\indicesof{a}{\aconst{v}}$ and
  $|\indicesof{a}{\aconst{v}}|$ remain unchanged in $C$, meaning that
  $\indicesof{a}{\aconst{v}}$ is well defined and $|\indicesof{a}{\aconst{v}}|$
  is finite in $C$.
  On the other hand, suppose $\indicesof{a}{\aconst{v}}=\emptypair$ in $C'$.  By
  $(iv)$, we have $\propath{a}{\aconst{v}}=\emptypair$ in $C'$ as well. If
  $\propath{a}{\aconst{v}}=\emptypair$ still in $C$, then it is clear that
  $\indicesof{a}{\aconst{v}}=\emptypair$ in $C$ as well, and so it is well
  defined and $|\indicesof{a}{\aconst{v}}|=0$.  The last case is when $\propath{a}{\aconst{v}}=\emptypair$ in $C'$ but
  $\propath{a}{\aconst{v}}\not=\emptypair$ in $C$.  This happens when the
  rule newly defines $\propath{a}{\aconst{v}}$.

  Now, if
  $a=\aconst{v}$, then $\indicesof{a}{\aconst{v}}$ is trivially well defined
  and $|\indicesof{a}{\aconst{v}}|=1$.  Otherwise, the rule newly defines
  $\propath{a}{\aconst{v}}$ to be some $(r,b)$, where $\propath{b}{\aconst{v}}\not=\emptypair$
  in $C'$. By ($iv$) and the induction hypothesis, $\indicesof{b}{\aconst{v}}\not=\emptypair$ in $C'$, so by
  ($i$), $\indicesof{b}{\aconst{v}}$ and $|\indicesof{b}{\aconst{v}}|$ remain
  unchanged from $C'$ to $C$.
  But now, by definition, we either have
  $\indicesof{a}{\aconst{v}}=\indicesof{b}{\aconst{v}}\cup\{j\}$ for some $j$ or
  $\indicesof{a}{\aconst{v}}=\indicesof{b}{\aconst{v}}$,
  which in either case is well defined.  Similarly, we have, in either case,
  $|\indicesof{a}{\aconst{v}}|=1+|\indicesof{b}{\aconst{v}}|$, which is finite.

  Finally, we show ($iv$).  The if direction follows directly by the
  definition of $\indicesof{a}{\aconst{v}}$.  For the only-if direction, consider the
  contrapositive.  Suppose that $\propath{a}{\aconst{v}}\not=\emptypair$.  Then by
  definition, either $\indicesof{a}{\aconst{v}}=\emptyset$,
  $\indicesof{a}{\aconst{v}}=\indicesof{b}{\aconst{v}}\cup\{j\}$ for some
  $b,j$, or $\indicesof{a}{\aconst{v}}=\indicesof{c}{\aconst{v}}$ for some
  $c$.  In each case, we know that it is well defined by ($ii$), so in
  each case, $\indicesof{a}{\aconst{v}}\not=\emptypair$.
  \qed
\end{proof}

\begin{lemma}\label{lem:cinvariant}
  Let $D$ be a derivation ending in configuration $C$, with $C\not=\unsat$.  The following is true
  in $C$ for all $a,v$:
\[\text{If }\propath{a}{\aconst{v}}\not=\emptypair\text{, then }
\theoryarr \models \reasonof{a}{\aconst{v}}\implies \forall\,i.\hspace*{-1em}\bigvee_{k\in\indicesof{a}{\aconst{v}}}\hspace*{-1em}\teq{i}{k} \vee \teq{\aread{a}{i}}{v}.\]
\end{lemma}
\begin{proof}
  The proof is by induction on the length of a derivation. If $C$ is the
  initial configuration, $\propathf$ is
  empty, so the statement holds vacuously.

  Now, suppose $C$ is obtained by applying a rule other than \ruleconf to configuration $C'$. We
  consider each rule and show that it preserves the property.

  First, if the rule is a conflict rule, then $\propathf$ is empty in $C$, so
  again, the statement holds vacuously.
  
  Second, if $\propath{a}{\aconst{v}}\not=\emptypair$ in $C'$, then
  by~\Cref{lem:reasonprops} and~\Cref{lem:indicesprops},
  $\reasonof{a}{\aconst{v}}$ and $\indicesof{a}{\aconst{v}}$ have the same value in $C$ and
  $C'$, so, since the property holds in $C'$ by the induction hypothesis, the
  property must hold also in $C$.

  The remaining case to consider is when $\propath{a}{\aconst{v}}$ is newly defined
  by the rule.  We consider all such rules.

  First note that the \ruleinitc rule defines $\propathf$ for
  a pair of arguments of the form $(\aconst{v},\aconst{v})$.  Since
  $\theoryarr\models\teq{\aconst{v}[i]}{v}$ for every $i$,
  the property clearly holds in $C$.

  Next, consider \rulecowd.  The newly defined value of $\propath{a}{\aconst{v}}$ is
  $(\btrue,\awrite{a}{j}{u})$.  We also know that
  $\propath{\awrite{a}{j}{u}}{\aconst{v}} \ne \emptypair$ in $C'$,
  so we have $\theoryarr \models \reasonof{\awrite{a}{j}{u}}{\aconst{v}}
  \implies
  \forall\,i.\bigvee_{k\in\indicesof{\awrite{a}{j}{u}}{\aconst{v}}}\teq{i}{k}
  \vee \teq{\aread{\awrite{a}{j}{u}}{i}}{v}$ in $C'$ by the 
  induction hypothesis. The same holds also in $C$ by~\Cref{lem:reasonprops,lem:indicesprops}.
  By the definition of $\reasonoff$, we have (in $C$)
  $\reasonof{a}{\aconst{v}}=\reasonof{\awrite{a}{j}{u}}{\aconst{v}}\wedge\btrue$.
  And, by the definition of $\indices$, we have (in $C$)
  $\indicesof{a}{\aconst{v}}=\indicesof{\awrite{a}{j}{u}}{\aconst{v}}\cup\{j\}$.
  Now, to show the property holds in $C$, let $\I$ be a $\theoryarr$-interpretation with
  $\I\models\reasonof{a}{\aconst{v}}$.  We must show that $\I\models
  \forall\,i.\bigvee_{k\in\indicesof{a}{\aconst{v}}}\teq{i}{k}
  \vee \teq{\aread{a}{i}}{v}$.
  To this end, let $\iota\in\sigma^{\I}$ be such that $\iota\ne k^{\I}$ for
  every $k\in\indicesof{a}{\aconst{v}}$.  We must show that $a^{\I}(\iota)=v^{\I}$.
  First, note that by the definition of $\reasonof{a}{\aconst{v}}$, we have
  $\I\models\reasonof{\awrite{a}{j}{u}}{\aconst{v}}$, so by the induction
  hypothesis, we also have
  $\I\models\forall\,i.\bigvee_{k\in\indicesof{a}{\aconst{v}}}\teq{i}{k}
  \vee \teq{\aread{\awrite{a}{j}{u}}{i}}{v}$.
  Since $\indicesof{\awrite{a}{j}{u}}{\aconst{v}}$ is a subset of
  $\indicesof{a}{\aconst{v}}$, it must be the case that
  $\awrite{a}{j}{u}^{\I}(\iota)=v^{\I}$.
  But we also know that $j^{\I}\ne \iota$ because
  $j\in\indicesof{a}{\aconst{v}}$, so by array axiom \eqref{eq:rowne}, we have
  $a^{\I}(\iota)=v^{\I}$.
  
  The argument for \rulecowu is similar, but we include it for completeness.
  Let $a=\awrite{a'}{j}{u}$. The newly defined value of $\propath{a}{\aconst{v}}$ is
  $(\btrue,a')$.  We also know that
  $\propath{a'}{\aconst{v}} \ne \emptypair$ in $C'$,
  so we have $\theoryarr \models \reasonof{a'}{\aconst{v}}
  \implies
  \forall\,i.\bigvee_{k\in\indicesof{a'}{\aconst{v}}}\teq{i}{k}
  \vee \teq{\aread{a'}{i}}{v}$ in $C'$ by the 
  induction hypothesis. The same holds also in $C$ by~\Cref{lem:reasonprops,lem:indicesprops}.
  By the definition of $\reasonoff$, we have (in $C$)
  $\reasonof{a}{\aconst{v}}=\reasonof{a'}{\aconst{v}}\wedge\btrue$.
  And, by the definition of $\indices$, we have (in $C$)
  $\indicesof{a}{\aconst{v}}=\indicesof{a'}{\aconst{v}}\cup\{j\}$.
  Now, to show the property holds in $C$, let $\I$ be a $\theoryarr$-interpretation with
  $\I\models\reasonof{a}{\aconst{v}}$.  We must show that $\I\models
  \forall\,i.\bigvee_{k\in\indicesof{a}{\aconst{v}}}\teq{i}{k}
  \vee \teq{\aread{a}{i}}{v}$.
  To this end, let $\iota\in\sigma^{\I}$ be such that $\iota\ne k^{\I}$ for
  every $k\in\indicesof{a}{\aconst{v}}$.  We must show that $a^{\I}(\iota)=v^{\I}$.
  First, note that by the definition of $\reasonof{a}{\aconst{v}}$, we have
  $\I\models\reasonof{a'}{\aconst{v}}$, so by the induction
  hypothesis, we also have
  $\I\models\forall\,i.\bigvee_{k\in\indicesof{a'}{\aconst{v}}}\teq{i}{k}
  \vee \teq{\aread{a'}{i}}{v}$.
  Since $\indicesof{a'}{\aconst{v}}$ is a subset of
  $\indicesof{a}{\aconst{v}}$, it must be the case that
  $(a')^{\I}(\iota)=v^{\I}$.
  But we also know that $j^{\I}\ne \iota$ because
  $j\in\indicesof{a}{\aconst{v}}$, so by array axiom \eqref{eq:rowne}, we have
  $a^{\I}(\iota)=v^{\I}$.

  Consider now \ruleceql.  The newly defined value of
  $\propath{a}{\aread{b}{i}}$ in $C$ is
  $(\teq{a}{c},c)$.  We also know that
  $\propath{c}{\aconst{v}} \ne \emptypair$ in $C'$,
  so we also have $\theoryarr \models \reasonof{c}{\aconst{v}}
  \implies
  \forall\,i.\bigvee_{k\in\indicesof{c}{\aconst{v}}}\teq{i}{k}
  \vee \teq{\aread{c}{i}}{v}$ in $C'$ by the 
  induction hypothesis. The same holds also in $C$ by~\Cref{lem:reasonprops,lem:indicesprops}.
  Now, by the definition of $\reasonoff$, we have (in $C$)
  $\reasonof{a}{\aconst{v}}=\reasonof{c}{\aconst{v}}\wedge\teq{a}{c}$.
  And, by the definition of $\indices$, we have (in $C$)
  $\indicesof{a}{\aconst{v}}=\indicesof{c}{\aconst{v}}$.
  And by equality reasoning, we have $\theoryarr \models \teq{a}{c} \implies
  \teq{\aread{a}{i}}{\aread{c}{i}}$.  It follows that
  $\theoryarr \models \reasonof{a}{\aconst{v}}
  \implies
  \forall\,i.\bigvee_{k\in\indicesof{a}{\aconst{v}}}\teq{i}{k}
  \vee \teq{\aread{a}{i}}{v}$
  The \ruleceqr case is analogous.

  These are the only rules that add a newly defined value to $\propathf$ of the
  form $(a,\aconst{v})$.
  \qed
\end{proof}

Refutational soundness is stated by the following theorem.

\begin{theorem}[Refutational Soundness]
  If a derivation starting with an initial configuration \config{\assertions}{\I}{\propathf}
  ends in \unsat, then $\assertions$ is $\theoryarr$-unsatisfiable.
\end{theorem}
\begin{proof}
  We first note that the only rule that derives \unsat is the \ruleconf
  rule, which applies when $\assertions \cup
  \{\teq{\aread{\awrite{a}{i}{u}}{i}}{u} \;|\; \awrite{a}{i}{u} \in
  \terms{\assertions}\}$ is unsatisfiable (in the empty theory).  Clearly, this
  implies that $\assertions \cup
  \{\teq{\aread{\awrite{a}{i}{u}}{i}}{u} \;|\; \awrite{a}{i}{u} \in
  \terms{\assertions}\}$ is also $\theoryarr$-unsatisfiable which is true iff
  $\assertions$ is $\theoryarr$-unsatisfiable, since
  $\{\teq{\aread{\awrite{a}{i}{u}}{i}}{u} \;|\; \awrite{a}{i}{u} \in
  \terms{\assertions}\}$ always holds in $\theoryarr$.

  It thus suffices to show that the remaining rules are \emph{satisfiability
  preserving} with respect to $\theoryarr$, that is,
  $\assertions$ is $\theoryarr$-satisfiable iff $\assertions'$ is $\theoryarr$-satisfiable, where
  $\assertions$ and $\assertions'$ are the values of $\assertions$ before and
  after applying the rule, respectively.  For most rules, this is trivially
  true, as they do not modify $\assertions$.

  Consider now rule \rulecongr, which adds $\reasonof{a}{\aread{b}{i}} \wedge
  \reasonof{a}{\aread{c}{k}} \wedge \teq{i}{k} \implies
  \teq{\aread{b}{i}}{\aread{c}{k}}$ to $\assertions$.  By~\Cref{lem:invariant},
  we know that $\theoryarr \models \reasonof{a}{\aread{b}{i}} \implies
  \teq{\aread{a}{i}}{\aread{b}{i}}$ and also $\theoryarr \models \reasonof{a}{\aread{c}{k}} \implies
  \teq{\aread{a}{k}}{\aread{c}{k}}$.  It follows that $\theoryarr \models
  \reasonof{a}{\aread{b}{i}} \wedge \reasonof{a}{\aread{c}{k}} \wedge
  \teq{i}{k} \implies \teq{\aread{b}{i}}{\aread{c}{k}}$.  Since the added
  formula is \emph{valid} in $\theoryarr$,
  adding it to $\assertions$ does not change its $\theoryarr$-satisfiability.

  Next, consider rule \rulediseq, which adds the formula $\tneq{a}{c} \implies
  \tneq{\aread{a}{\adiff{a}{c}}}{\aread{c}{\adiff{a}{c}}}$.  Call the result
  $\assertions'$.  Clearly, if $\assertions'$ is $\theoryarr$-satisfiable, then
  so is $\assertions$.  In the other direction, suppose $\assertions$ is
  $\theoryarr$-satisfiable with interpretation $\I$.  If $a^\I = c^\I$, then
  $\assertions'$ is also satisfied by $\I$.  Otherwise, if $\I \models
  \tneq{a}{c}$, then $a^\I(\iota) \not= c^\I(\iota)$ for some
  $\iota$.  We can then modify $\I$ to $\I'$ where $\I'$ is identical to $\I$
  except that it interprets $k_{a,c}$ as $\iota$.  $\I'$ then satisfies
  $\assertions'$.

  We next consider rule \ruleroc, which adds
  $\reasonof{\aconst{v}}{\aread{b}{i}}\implies\teq{\aread{b}{i}}{v}$.
  By~\Cref{lem:invariant}, we know that $\theoryarr \models \reasonof{\aconst{v}}{\aread{b}{i}} \implies
  \teq{\aread{\aconst{v}}{i}}{\aread{b}{i}}$.  But by array
  axiom~\eqref{eq:roc}, we know that $\theoryarr\models \aread{\aconst{v}}{i}=v$, so
  $\theoryarr \models \reasonof{\aconst{v}}{\aread{b}{i}} \implies
  \teq{\aread{b}{i}}{v}$. Since the added
  formula is \emph{valid} in $\theoryarr$,
  adding it to $\assertions$ does not change its $\theoryarr$-satisfiability.
  
  Finally, consider rule \rulecongc, which adds
  $\reasonof{a}{\aconst{v}}\wedge\reasonof{a}{\aconst{w}}\wedge\linebreak
    \bexists{i : \indexsort}\bigwedge_{k \in \indicesof{a}{\aconst{v}}\cup\indicesof{a}{\aconst{w}}}\tneq{i}{k}
    \implies \teq{v}{w}$.
  By~\Cref{lem:cinvariant},
  we know that
  $\theoryarr \models \reasonof{a}{\aconst{v}}\implies \forall\,i.\bigvee_{k\in\indicesof{a}{\aconst{v}}}\teq{i}{k} \vee \teq{\aread{a}{i}}{v}$
  and also
  $\theoryarr \models \reasonof{a}{\aconst{w}}\implies \forall\,i.\bigvee_{k\in\indicesof{a}{\aconst{w}}}\teq{i}{k} \vee \teq{\aread{a}{i}}{w}$.
  It follows that
  $\reasonof{a}{\aconst{v}}\wedge\reasonof{a}{\aconst{w}}\wedge
    \bexists{i : \indexsort}\bigwedge_{k \in \indicesof{a}{\aconst{v}}\cup\indicesof{a}{\aconst{w}}}\tneq{i}{k}
    \implies \teq{v}{w}$.
  Since the added
  formula is \emph{valid} in $\theoryarr$,
  adding it to $\assertions$ does not change its $\theoryarr$-satisfiability.

  These are all of the rules that modify $\assertions$.
  \qed
\end{proof}

Refutational soundness ensures that the calculus never reports \unsat unless
the initial set of constraints is unsatisfiable.  Satisfiability soundness is
the complementary property: if the calculus \emph{does not} report \unsat,
then the initial set of constraints is satisfiable.\footnote{This is often
referred to as completeness, but we find the term satisfiability soundness more
precise in this context.}
We need one more lemma.
\begin{lemma}\label{lem:interpreasons}
  Let $D$ be a derivation ending in configuration $C$, with $C\not=\unsat$.
  The following is true in $C$ for all $a,t$.  If $\I\ne\I_0$ and
  $\reasonof{a}{t}\ne\emptypair$, then $\I \models \reasonof{a}{t}$.
\end{lemma}
\begin{proof}
  The proof is by induction on the length of a derivation. If $C$ is the
  initial configuration, $\I=\I_0$, so the statement holds vacuously.

  Now, suppose $C$ is obtained by applying a rule to configuration $C'$. We
  consider each rule and show that it preserves the property.

  First, if the rule is a conflict rule, then $\propathf$ is empty in $C$, so
  again, the statement holds vacuously.
  The remaining case to consider is when $\propath{a}{t}$ is newly defined
  by the rule.  We consider all such rules.

  First note that the \ruleinitr, \ruleinitw, and \ruleinitc rules define $\propathf$ for
  a pair of arguments of the form $(a,a[i])$ or $(a,a)$.  Since
  $\reasonof{a}{t}=\btrue$ in these cases, the property holds.

  For all other rules, if $\propath{a}{t}$ is newly set to $(r,c)$ in $C$, we
  have $\reasonof{a}{t}=\reasonof{c}{t}\wedge r$.  By the induction hypothesis
  and \Cref{lem:reasonprops}, we have $\I\models\reasonof{c}{t}$, so we simply
  have to show $\I\models r$.  But in every such rule, either $r=\btrue$ or
  $\I\models r$ is a premise.
  \qed
\end{proof}


\begin{theorem}[Satisfiability Soundness]
  If a derivation starting with an initial configuration \config{\assertions}{\I_0}{\propathf_0}
  ends in a configuration other than \unsat, and every applicable rule
  is redundant, then $\assertions$ is $\theoryarr$-satisfiable.
\end{theorem}
\begin{proof}

We first note that because \ruleconf does not apply, it must be the case that
$\assertions \cup \{\teq{\aread{\awrite{a}{i}{u}}{i}}{u} \;|\; \awrite{a}{i}{u}
\in \terms{\assertions}\}$ is satisfiable in the empty theory.  Moreover,
because \ruleinterp does not apply, a specific satisfying interpretation $\I$
is already stored in the final configuration.


We will show that $\assertions$ is
$\theoryarr$-satisfiable by using $\I$ to construct a $\theoryarr$-interpretation $\Iarr$
that satisfies $\assertions$. We start by defining the domains of $\Iarr$ as
\begin{align*}
\Iarr_\indexsort &= \I_\indexsort\enspace,\\
\Iarr_\elementsort &= \I_\elementsort\enspace,\\
\Iarr_{\arraysort} &= \lbrace\;\ f \;\;|\;\; f : \Iarr_\indexsort \mapsto \Iarr_\elementsort \;\rbrace\enspace.
\end{align*}

For this proof, it will be convenient to use the SMT-LIB function names for array reads
and writes.  In the rest of the proof, we consider $\aread{a}{i}$ to be
syntactic sugar for $\select(a,i)$ and $\awrite{a}{i}{u}$ to be syntactic sugar
for $\store(a,i,u)$.  We also introduce $\const(v)$ and consider
$\aconst{v}$ syntactic sugar for $\const(v)$.
We define these function symbols in $\Iarr$ as follows:
\begin{align*}
\select^\Iarr &= \lambda\,\alpha:\Iarr_{\arraysort}.\:\lambda\,\iota:\Iarr_\indexsort.\:\alpha(\iota)\enspace,\\
\store^\Iarr &= \lambda\,\alpha:\Iarr_{\arraysort}.\:\lambda\,\iota:\Iarr_\indexsort.\:\lambda\,\epsilon:\Iarr_\elementsort.\:(\lambda\,\kappa:\Iarr_\indexsort. \text{if
} \iota = \kappa \text{ then } \epsilon \text{ else } \alpha(\kappa))\enspace.\\
\const^\Iarr &= \lambda\,\epsilon:\Iarr_\elementsort.\:(\lambda\,\kappa:\Iarr_\indexsort.\epsilon)
\end{align*}

For all variables $v$ that are not of sort $\arraysort$, we keep their original interpretations in $\I$, i.e., we take $v^\Iarr = v^\I$.

Let $\epsilon$ be some element of $\Iarr_\elementsort$.  We interpret each array $a:\arraysort$ as the corresponding function from $\I$, but in 
a restricted manner.

\begin{align*}
a^\Iarr = \lambda\,\iota:\Iarr_\indexsort.\:
\begin{cases}
\select^\I(b^\I,i^\I) & \text{if } \propath{a}{\aread{b}{i}} \ne\emptypair
\text{ and } i^\I = \iota\enspace \\
v^\I & \text{if } \propath{a}{\aconst{v}}\ne\emptypair\text{ and
}\iota\ne k^\I \text{ for every }k\in\indicesof{a}{\aconst{v}}\\
\epsilon & \text{otherwise}
\end{cases}
\end{align*}
Intuitively, we interpret the elements according to the corresponding
interpretations of the reads, we ensure that default elements match with any
relevant constant arrays, and we otherwise pick an arbitrary default element.

We first show that the interpretation is well-defined.
To show that the first case above is well defined, suppose that for some variable
$a$, we have both $\propath{a}{\aread{b}{i}}\ne\emptypair$ and
$\propath{a}{\aread{c}{k}}\ne\emptypair$ with $i^\I=k^\I=\iota$, and suppose
towards a contradiction
that $\select^\I(b^\I,i^\I)\ne\select^\I(c^\I,k^\I)$, or, in other words,
$\I\models \tneq{\aread{b}{i}}{\aread{c}{k}}$.  Then, the premises of
rule \rulecongr are satisfied, so in order for it to be redundant, it must mean that
$\reasonof{a}{\aread{b}{i}} \wedge \reasonof{a}{\aread{c}{k}} \wedge \teq{i}{k} \implies \teq{\aread{b}{i}}{\aread{c}{k}}\in\assertions$.
Now, by~\Cref{lem:interpreasons}, we know that
$\I\models\reasonof{a}{\aread{b}{i}}$ and $\I\models\reasonof{a}{\aread{c}{k}}$.
We also know that $\I\models \teq{i}{k}$.  But this means that we must have
$\I\models\teq{\aread{b}{i}}{\aread{c}{k}}$, which contradicts our assumption.

To show that the second case is well defined, suppose that for some variable
$a$, we have both $\propath{a}{\aconst{v}}\ne\emptypair$ and
$\propath{a}{\aconst{w}}\ne\emptypair$, and suppose towards a contradiction
that for some element $\iota : \indexsort$, we have $\iota\ne k^\I$ for every
$k\in\indicesof{a}{\aconst{v}}\cup\indicesof{a}{\aconst{w}}$ and $w^\I\ne v^\I$.
Then, the premises of
rule \rulecongc are satisfied, so in order for it to be redundant, it must mean that
$\reasonof{a}{\aconst{v}} \wedge \reasonof{a}{\aconst{w}} \wedge
\bexists{i : \indexsort}\bigwedge_{k \in \indicesof{a}{\aconst{v}}\cup\indicesof{a}{\aconst{w}}}\tneq{i}{k}
    \implies \teq{v}{w} \in \assertions$.
Now, by~\Cref{lem:interpreasons}, we know that
$\I\models\reasonof{a}{\aconst{v}}$ and $\I\models\reasonof{a}{\aconst{w}}$.
But this means that we must have
$\I\models\teq{v}{w}$, which contradicts our assumption.

We must also show that the first and second cases don't conflict.  Suppose
towards a contradiction that for some $\iota : \indexsort$, we have both
$\propath{a}{\aread{b}{i}}\ne\emptypair$, with $i^\I=\iota$ and
$\propath{a}{\aconst{v}}\ne\emptypair$ with $\iota\ne k^\I$ for every
$k\in\indicesof{a}{\aconst{v}}$ and $\select^\I(b^\I,i^\I)\ne v^\I$, or in
other words, $\I\models \tneq{\aread{b}{i}}{v}$.
It is not hard to see that in this case, we must also have
$\propath{\aconst{v}}{\aread{b}{i}}\ne\emptypair$. 
Then, the premises of
rule \ruleroc are satisfied, so in order for it to be redundant, it must mean that
$\reasonof{\aconst{v}}{\aread{b}{i}} \implies \teq{\aread{b}{i}}{v} \in \assertions$.
Now, by~\Cref{lem:interpreasons}, we know that
$\I\models\reasonof{\aconst{v}}{\aread{b}{i}}$
But this means that we must have
$\I\models\teq{\aread{b}{i}}{v}$, which contradicts our assumption.

It is easy to see that the definitions of $\select$, $\store$, and $\const$ satisfy the
axioms of $\theoryarr$.  Now, we proceed to show that $\Iarr \models \assertions$.
First, note that by construction, for a variable $v$
not of sort $\arraysort$, we have $v^\Iarr = v^\I$.  Next, consider a
term of the form $\aread{a}{i}\in\terms{\assertions}$.  We
know by the fact that rule \ruleinitr is redundant that
$\propath{a}{\aread{a}{i}}\ne\emptypair$.
Then, by definition, $a^\Iarr(i^\Iarr)=\select^\I(a^\I,i^\I)$.  Then, by
the definition of $\select^\Iarr$, it's clear that $\select^\Iarr(a^\Iarr,i^\Iarr)=\select^\I(a^\I,i^\I)$.
Since a flat term of a non-array sort can only be a variable or an array read,
this shows that if we have flat terms $s,t\in\terms{\assertions}$ of a non-array sort,
$s^\Iarr = s^\I$ and $t^\Iarr=t^\I$.  Thus, all equalities $\teq{s}{t}$ and
disequalities $\tneq{s}{t}$
between non-array terms must have the same truth value in $\Iarr$ as in $\I$.

Next, consider equalities and disequalities between variables of sort $\arraysort$.
For a disequality $\tneq{a}{c} \in \assertions$, we know by rule \rulediseq that 
$\tneq{a}{c} \implies
\tneq{\aread{a}{\adiff{a}{c}}}{\aread{c}{\adiff{a}{c}}}\in\assertions$.
So, since $\I\models\tneq{a}{c}$, we also have $\I\models
\tneq{\aread{a}{\adiff{a}{c}}}{\aread{c}{\adiff{a}{c}}}$.
By \ruleinitr, we have $\propath{a}{\aread{a}{\adiff{a}{c}}}\ne\emptypair$ and
$\propath{c}{\aread{c}{\adiff{a}{c}}}\ne\emptypair$, so
by the definition of $a^\Iarr$ and $c^\Iarr$, it's clear that
$a^\Iarr(\adiff{a}{c}^\Iarr)=(\aread{a}{\adiff{a}{c}})^\I \ne
(\aread{c}{\adiff{a}{c}})^\I = c^\Iarr(\adiff{a}{c}^\Iarr)$.  Therefore,
$a^\Iarr\ne c^\Iarr$.

To see that equalities $\teq{a}{c}$ are
satisfied, let $\iota : \indexsort$ and consider $a^\Iarr(\iota)$.  If
$\propath{a}{\aread{b}{i}}\ne\emptypair$ for some $b,i$ with $i^\Iarr=\iota$,
then $a^\Iarr(\iota)=\select^\I(b^\I,i^\I)$.  Now, we know by the fact that
rule \ruleeqr is redundant that we must also have
$\propath{c}{\aread{b}{i}}\ne\emptypair$.  It follows then that also
$c^\Iarr(\iota)=\select^\I(b^\I,i^\I)$.

Suppose now that $\propath{a}{\aread{b}{i}}=\emptypair$ for every $b,i$ with
$i^\Iarr=\iota$.  Suppose further that $\propath{a}{\aconst{v}}\ne\emptypair$ and
$\iota\ne k^\I$ for every $k\in\indicesof{a}{\aconst{v}}$ so that
$a^\Iarr(\iota)=v^\I$.
We know that rule \ruleceqr is redundant, so we must have
$\propath{c}{\aconst{v}}\ne\emptypair$.  If $\iota\ne k^\I$ for every
$k\in\indicesof{c}{\aconst{v}}$, then we also have $c^\Iarr(\iota)=v^\I$ as
required.  However, if $\iota=k^\I$ for some $k\in\indicesof{c}{\aconst{v}}$,
then there must be a write to index $k$.  And the corresponding read introduced by
\ruleinitw would have propagated to $a$ contradicting our assumption.
A symmetric argument can be made starting with $c^\Iarr(\iota)$.
In the final case, where none of the above hold, we have $a^\Iarr(\iota) =
\epsilon = c^\Iarr(\iota)$.
This concludes the case for equalities and disequalities between variables of
sort $\arraysort$.

Next, consider an equality of the form $\teq{c}{\awrite{a}{i}{u}}$.
Let $\iota : \indexsort$ and consider $c^\Iarr(\iota)$.
First, suppose $\iota = i^\Iarr$.  We know that
$\propath{\awrite{a}{i}{u}}{\aread{\awrite{a}{i}{u}}{i}}\ne\emptypair$ by
\ruleinitw.  It follows from \ruleeql that
$\propath{c}{\aread{\awrite{a}{i}{u}}{i}}\ne\emptypair$.
Thus, by definition, $c^\Iarr(\iota)=(\aread{\awrite{a}{i}{u}}{i})^\I$.  But we
know by \ruleinterp that $\I\models\teq{\aread{\awrite{a}{i}{u}}{i}}{u}$, so
$c^\Iarr=u^\I$.  And by definition of $\select^\Iarr$ and $\store^\Iarr$, we also have
$(\aread{\awrite{a}{i}{u}}{i})^\Iarr=u^\Iarr=u^\I$.

Otherwise, by the definition of
$\store$, and because $i^\Iarr\ne \iota$, we have
$(\awrite{a}{i}{u})^\Iarr(\iota)=a^\Iarr(\iota)$.  So, we must show that $c^\Iarr(\iota)=a^\Iarr(\iota)$.
If
$\propath{c}{\aread{b}{k}}\ne\emptypair$ for some $b,k$ with $k^\Iarr=\iota$,
then $c^\Iarr(\iota)=\select^\I(b^\I,k^\I)$.  Now, we know by the fact that
rule \ruleeqr cannot be applied that we must also have
$\propath{\awrite{a}{i}{u}}{\aread{b}{k}}\ne\emptypair$.
Then, because
$i^\I=i^\Iarr\ne k^\Iarr=k^\I$, we also have by \rulerowd that $\propath{a}{\aread{b}{k}}\ne\emptypair$.
So, $a^\Iarr(\iota)=\select^\I(b^\I,k^\I)$.

Next, suppose that $\propath{c}{\aread{b}{k}}=\emptypair$ for every
$k^\I=\iota$ but $\propath{c}{\aconst{v}}\ne\emptypair$ for some $v$ and 
$\iota\ne k^\I$ for every $k\in\indicesof{c}{\aconst{v}}$ so that
$c^\Iarr(\iota)=v^\I$.
Now, by \ruleceqr, we know
that $\propath{\awrite{a}{i}{u}}{\aconst{v}}\ne\emptypair$.
But if $\iota=k^\I$ for some $k\in\indicesof{\awrite{a}{i}{u}}$,
then there must be a write to index $k$.  And the corresponding read introduced by
\ruleinitw would have propagated to $c$ contradicting our assumption.
Thus, we must have $\iota\not=k^\I$ for every
$k\in\indicesof{\awrite{a}{i}{u}}$.  It follows that $a^\Iarr(\iota)=v^\I$.
A symmetric argument can be made with $\awrite{a}{i}{u}$.
In the final case, where none of the above hold, we have $a^\Iarr(\iota) =
\epsilon = c^\Iarr(\iota)$.

Finally, consider an equality of the form $\teq{a}{\aconst{v}}$.
We know by
\ruleinitc that $\propath{\aconst{v}}{\aconst{v}}\ne\emptypair$.  So, by
\ruleceql, we also have $\propath{a}{\aconst{v}}\ne\emptypair$.  Consider
$a^\Iarr(\iota)$.
We know that $\aconst{v}^\Iarr(\iota)=v^\Iarr=v^\I$.
If
$\propath{a}{\aread{b}{k}}\ne\emptypair$ for some $b,k$ with $k^\Iarr=\iota$,
then $a^\Iarr(\iota)=(\aread{b}{k})^\I$.  But by \ruleeqr, we also then
have $\propath{\aconst{v}}{\aread{b}{k}}\ne\emptypair$.  Suppose
$\I\models\aread{b}{k}\ne v$.
Then, the premises of
rule \ruleroc are satisfied, so we must have
$\reasonof{\aconst{v}}{\aread{b}{k}} \implies \teq{\aread{b}{k}}{v} \in \assertions$.
Now, by~\Cref{lem:interpreasons}, we know that
$\I\models\reasonof{\aconst{v}}{\aread{b}{k}}$
But this means that we must have
$\I\models\teq{\aread{b}{k}}{v}$, which contradicts our assumption. Thus,
$(\aread{b}{k})^\I=v^\I$.
On the other hand, suppose that $\propath{a}{\aread{b}{k}}=\emptypair$ for every
$k^\I=\iota$.
Suppose further that $\propath{a}{\aconst{w}}\ne\emptypair$ and $\iota\ne k^\I$
for every $k\in\indicesof{a}{\aconst{w}}$ so that $a^\Iarr(\iota)=w^\I$.
We know from \ruleceqr that $\propath{\aconst{v}}{\aconst{w}}\ne\emptypair$.
But then if $v^\I\ne w^\I$, \rulecongc would apply, so we must have
$v^\I=w^\I$.
A symmetric argument starting with $\aconst{v}$ shows that it must always be
the case that $a^\Iarr(\iota)=v^\I$.
\qed
\end{proof}
\end{arxiv}

\end{document}

%% file: macros.tex
\newcommand{\rem}[1]{\textcolor{red}{[#1]}}
\newcommand{\an}[1]{\rem{#1 --an}}
\newcommand{\cb}[1]{\rem{#1 --cb}}
\newcommand{\mpr}[1]{\rem{#1 --mp}}

\newtheorem{assumption}{Assumption}{\bfseries}{\itshape}
\Crefname{assumption}{\text{Assumption}}{\text{Assumptions}}
\crefname{assumption}{\text{Assumption}}{\text{Assumptions}}

\newcommand{\teqf}{\ensuremath{\approx}\xspace}
\newcommand{\tneqf}{\ensuremath{\not\teqf}}
\newcommand{\teq}[2]{\ensuremath{#1 \teqf #2}\xspace}
\newcommand{\tneq}[2]{\ensuremath{#1 \tneqf #2}\xspace}

\newcommand{\areadf}{\ensuremath{\mathit{read}}}
\newcommand{\aread}[2]{\ensuremath{#1[#2]}\xspace}
\newcommand{\awritef}{\ensuremath{\mathit{write}}}
\newcommand{\awrite}[3]{\ensuremath{#1\langle#2 \triangleleft #3\rangle}\xspace}
\newcommand{\accf}{\ensuremath{\mathit{S}}\xspace}
\newcommand{\acc}[1]{\ensuremath{\accf(#1)}\xspace}
\newcommand{\aconstf}{\ensuremath{\mathit{const}}\xspace}
\newcommand{\aconst}[1]{\ensuremath{\langle#1\rangle}\xspace}
\newcommand{\termsf}{\ensuremath{\mathit{T}}\xspace}
\newcommand{\terms}[1]{\ensuremath{\termsf(#1)}\xspace}
\newcommand{\termsaf}{\ensuremath{\termsf_\mathcal{A}}\xspace}
\newcommand{\termsa}[1]{\ensuremath{\termsaf(#1)}\xspace}
\newcommand{\assertions}{\ensuremath{A}\xspace}
\newcommand{\config}[3]{\ensuremath{\langle #1, #2, #3\rangle}\xspace}
\newcommand{\conflictf}{\ensuremath{\mathcal{C}}\xspace}
\newcommand{\conflict}[6]{\ensuremath{\langle #1, #2, #3, #4, #5, #6\rangle_R}\xspace}
\newcommand{\cconflict}[5]{\ensuremath{\langle #1, #2, #3, #4, #5\rangle_C}\xspace}
\newcommand{\indices}{\ensuremath{I}\xspace}
\newcommand{\reasonof}[2]{\ensuremath{\mathcal{R}(#1,#2)}\xspace}
\newcommand{\reasonoff}{\ensuremath{\mathcal{R}}\xspace}
\newcommand{\indicesof}[2]{\ensuremath{\indices(#1,#2)}\xspace}

\newcommand{\bool}{\ensuremath{\mathsf{Bool}}\xspace}
\newcommand{\boolnotf}{\ensuremath{\neg}\xspace}
\newcommand{\boolandf}{\ensuremath{\wedge}\xspace}

\newcommand{\sort}{\ensuremath{s}\xspace}
\newcommand{\sorts}{\ensuremath{S}\xspace}

\newcommand{\vars}{\ensuremath{X}\xspace}
\newcommand{\varss}[1]{\ensuremath{\vars_{#1}}\xspace}

\newcommand{\sig}{\ensuremath{\Sigma}\xspace}
\newcommand{\sigs}{\ensuremath{\sig^s}\xspace}
\newcommand{\sigf}{\ensuremath{\sig^f}\xspace}

\newcommand{\I}{\ensuremath{{\mathcal{I}}}\xspace}
\newcommand{\Iarr}{\ensuremath{{\mathcal{A}}}\xspace}
\newcommand{\Is}{\ensuremath{{I}}\xspace}
\newcommand{\sorti}{\ensuremath{\sort^\I}\xspace}

\newcommand{\T}{\ensuremath{\mathcal{T}}\xspace}

\newcommand{\indexsort}{\ensuremath{\sigma}\xspace}
\newcommand{\elementsort}{\ensuremath{\tau}\xspace}
\newcommand{\arraysort}{\ensuremath{(\indexsort \rightarrow \elementsort)}\xspace}
\newcommand{\select}{\texttt{select}}
\newcommand{\store}{\texttt{store}}
\newcommand{\const}{\texttt{const}}

\newcommand{\theorybv}{\ensuremath{\T_{\mathcal{BV}}}\xspace}

\newcommand{\theoryeuf}{\ensuremath{\T_{\mathcal{E}}}\xspace}

\newcommand{\theoryarr}{\ensuremath{\T_\mathcal{A}}\xspace}
\newcommand{\sigarr}{\ensuremath{\sig_\mathcal{A}}\xspace}

\newcommand{\carr}{\ensuremath{\langle\mathcal{A}\rangle}\xspace}
\newcommand{\theorycarr}{\ensuremath{\T_{\carr}}\xspace}
\newcommand{\sigcarr}{\ensuremath{\sig_{\carr}}\xspace}
\newcommand{\Icarr}{\ensuremath{{\mathcal{I}_{\carr}}}\xspace}

\newcommand{\aconsti}[1]{\ensuremath{\langle#1\rangle_\indexsort}\xspace}
\newcommand{\adiff}[2]{\ensuremath{k_{\{#1,#2\}}}\xspace}

\newcommand{\btrue}{\ensuremath{\top}\xspace}
\newcommand{\bfalse}{\ensuremath{\bot}\xspace}
\newcommand{\pluseq}{\ensuremath{\mathrel{+}=}}

\newcommand{\emptymap}{\ensuremath{\lambda x,y \,.\, \emptypair}\xspace}
\newcommand{\emptypair}{\ensuremath{()}\xspace}
\newcommand{\domsize}[1]{\ensuremath{|#1^{\I}|}\xspace}
\newcommand{\propathf}{\ensuremath{\pi}\xspace}
\newcommand{\propath}[2]{\ensuremath{\propathf(#1,#2)}\xspace}

\newcommand{\imp}[2]{\ensuremath{#1\implies#2}\xspace}
\newcommand{\bimp}[2]{\ensuremath{#1\iff#2}\xspace}

\newcommand{\bforall}[2]{\ensuremath{\forall #1.\,#2}\xspace}
\newcommand{\bexists}[2]{\ensuremath{\exists #1.\,#2}\xspace}
\newcommand{\distinct}{\ensuremath{\mathit{distinct}}\xspace}
\newcommand{\distinctn}{\ensuremath{\mathit{distinct_N}}\xspace}

\newcommand{\rn}[1]{\textsf{\small #1}}
\newcommand{\ruleinterp}{\rn{Interp}\xspace}
\newcommand{\ruleinitr}{\rn{InitR}\xspace}
\newcommand{\ruleinitw}{\rn{InitW}\xspace}
\newcommand{\ruleinitc}{\rn{InitC}\xspace}
\newcommand{\rulerowd}{\rn{RowD}\xspace}
\newcommand{\rulerowu}{\rn{RowU}\xspace}
\newcommand{\rulecowd}{\rn{CowD}\xspace}
\newcommand{\rulecowu}{\rn{CowU}\xspace}
\newcommand{\ruleroc}{\rn{Roc}\xspace}
\newcommand{\ruleeqr}{\rn{EqR}\xspace}
\newcommand{\ruleeql}{\rn{EqL}\xspace}
\newcommand{\ruleceqr}{\rn{CEqR}\xspace}
\newcommand{\ruleceql}{\rn{CEqL}\xspace}
\newcommand{\rulediseq}{\rn{DisEq}\xspace}
\newcommand{\rulecongr}{\rn{CongR}\xspace}
\newcommand{\rulecongc}{\rn{CongC}\xspace}
\newcommand{\ruleconf}{\rn{Conf}\xspace}
\newcommand{\unsat}{\rn{unsat}\xspace}

\newcommand{\dpa}{DP\textsubscript{A}\xspace}
\newcommand{\none}{\ensuremath{\mathsf{none}}\xspace}
\newcommand{\calcbase}{\textsf{AEXT}\xspace}
\newcommand{\calcext}{\textsf{CAEXT}\xspace}

\newcommand{\configbzlabase}{\emph{Bitwuzla}\xspace}
\newcommand{\configbzlaconst}{\emph{Bitwu\-zla\textsubscript{\carr}}\xspace}
\newcommand{\configcvcv}{\emph{cvc5}\xspace}
\newcommand{\configziii}{\emph{Z3}\xspace}
\newcommand{\configmsat}{\emph{MathSAT5}\xspace}

\newcommand{\benchsmtlib}{\emph{smtlib}\xspace}
\newcommand{\benchmbqi}{\emph{mbqi}\xspace}
\newcommand{\benchpononb}{\emph{k\nobreakdash-ind}\xspace}
\newcommand{\benchpono}{\emph{k-ind}\xspace}
\newcommand{\benchhevm}{\emph{hevm}\xspace}
\newcommand{\benchcraft}{\emph{crafted}\xspace}
\newcommand{\benchcraftc}{\emph{crafted}\textsubscript{\carr}\xspace}
\newcommand{\benchcraftq}{\emph{crafted}\textsubscript{\ensuremath{\forall}}\xspace}


%% file: table/arrayops.tex
{%
  \begin{tabular}{l@{\hspace{2em}}l@{\hspace{2em}}l@{\hspace{2em}}l}
    \toprule
    \textbf{Syntax} & \textbf{SMT-LIB} & \textbf{Arity} & \textbf{Comment}\\
    \midrule
    \aread{\cdot}{\cdot}
    & \texttt{select}
    & $\arraysort \times \indexsort \rightarrow \elementsort$
    & Array access (read)
    \\
    \awrite{\cdot}{\cdot}{\cdot}
    & \texttt{store}
    & $\arraysort \times \indexsort \times \elementsort \rightarrow \arraysort$
    & Array update (write)
    \\
    $\aconst{\cdot}_\indexsort$
    &
    & $\elementsort \rightarrow \arraysort$
    & Constant array
    \\
    \bottomrule
  \end{tabular}
}

%% file: fig/example2.tex
\begin{tikzpicture}[
    >={Latex[length=2mm,width=2mm]},
    cell/.style={
      draw, minimum width=7mm, minimum height=7mm,
      inner sep=0pt, font=\footnotesize, anchor=center
    },
    cellT/.style={cell, fill=blue!15},
    cellF/.style={cell, fill=orange!22},
    update/.style={line width=1pt},
    arrname/.style={font=\small, anchor=center},
    storearr/.style={->, black, thick, shorten >=2pt, shorten <=2pt},
    eqarr/.style={thick},
    storelabel/.style={font=\footnotesize, fill=white, inner sep=1pt, text=black}
  ]

  %
  \def\cw{0.75}   
  \def\ch{0.7}    

  \def\yRoot   {-1.5}
  \def\yStore  { 1.0}
  \def\yEq     { 3}

  \def\xCV   {-4.0}     
  \def\xA    { 2.0}     
  \def\xCW   { 8.0}     

  \def\xOne  {-4.0}     
  \def\xTwo  { 0.0}     
  \def\xThree{ 4.0}     
  \def\xFour { 8.0}     


  \node[cellF] (cvA) at ({\xCV - 1.5*\cw}, \yRoot) {$\bot$};
  \node[cellF] (cvB) at ({\xCV - 0.5*\cw}, \yRoot) {$\bot$};
  \node[cellF] (cvC) at ({\xCV + 0.5*\cw}, \yRoot) {$\bot$};
  \node[cellF] (cvD) at ({\xCV + 1.5*\cw}, \yRoot) {$\bot$};
  \node[arrname] (cv_name) at (\xCV, \yRoot - 0.7) {$\aconst{v}$};
  \coordinate (cv_top) at (\xCV, {\yRoot + 0.4});

  \node[cellT] (aA) at ({\xA - 1.5*\cw}, \yRoot) {$\top$};
  \node[cellF] (aB) at ({\xA - 0.5*\cw}, \yRoot) {$\bot$};
  \node[cellT] (aC) at ({\xA + 0.5*\cw}, \yRoot) {$\top$};
  \node[cellF] (aD) at ({\xA + 1.5*\cw}, \yRoot) {$\bot$};
  \node[arrname] (a_name) at (\xA, \yRoot - 0.7) {$a$};
  \coordinate (a_top) at (\xA, {\yRoot + 0.4});

  \node[cellT] (cwA) at ({\xCW - 1.5*\cw}, \yRoot) {$\top$};
  \node[cellT] (cwB) at ({\xCW - 0.5*\cw}, \yRoot) {$\top$};
  \node[cellT] (cwC) at ({\xCW + 0.5*\cw}, \yRoot) {$\top$};
  \node[cellT] (cwD) at ({\xCW + 1.5*\cw}, \yRoot) {$\top$};
  \node[arrname] (cw_name) at (\xCW, \yRoot - 0.7) {$\aconst{w}$};
  \coordinate (cw_top) at (\xCW, {\yRoot + 0.4});


  \node[arrname] at (\xOne, \yStore + 0.7) {\awrite{\aconst{v}}{i_1}{u_1}};
  \node[cellT] (s1A) at ({\xOne - 1.5*\cw}, \yStore) {$\top$};
  \node[cellF]       at ({\xOne - 0.5*\cw}, \yStore) {$\bot$};
  \node[cellF]       at ({\xOne + 0.5*\cw}, \yStore) {$\bot$};
  \node[cellF]       at ({\xOne + 1.5*\cw}, \yStore) {$\bot$};
  \draw[update] (s1A.south west) rectangle (s1A.north east);
  \coordinate (s1_bot) at (\xOne, {\yStore - 0.35});
  \coordinate (s1_top) at (\xOne, {\yStore + 1});

  \node[arrname] at (\xTwo, \yStore + 0.7) {\awrite{a}{j_1}{u_2}};
  \node[cellT]       at ({\xTwo - 1.5*\cw}, \yStore) {$\top$};
  \node[cellF]       at ({\xTwo - 0.5*\cw}, \yStore) {$\bot$};
  \node[cellF] (s2C) at ({\xTwo + 0.5*\cw}, \yStore) {$\bot$};
  \node[cellF]       at ({\xTwo + 1.5*\cw}, \yStore) {$\bot$};
  \draw[update] (s2C.south west) rectangle (s2C.north east);
  \coordinate (s2_bot) at (\xTwo, {\yStore - 0.35});
  \coordinate (s2_top) at (\xTwo, {\yStore + 1});

  \node[arrname] at (\xThree, \yStore + 0.7) {\awrite{a}{i_2}{u_3}};
  \node[cellT]       at ({\xThree - 1.5*\cw}, \yStore) {$\top$};
  \node[cellT] (s3B) at ({\xThree - 0.5*\cw}, \yStore) {$\top$};
  \node[cellT]       at ({\xThree + 0.5*\cw}, \yStore) {$\top$};
  \node[cellF]       at ({\xThree + 1.5*\cw}, \yStore) {$\bot$};
  \draw[update] (s3B.south west) rectangle (s3B.north east);
  \coordinate (s3_bot) at (\xThree, {\yStore - 0.35});
  \coordinate (s3_top) at (\xThree, {\yStore + 1});

  \node[arrname] at (\xFour, \yStore + 0.7) {\awrite{\aconst{w}}{j_2}{u_4}};
  \node[cellT]       at ({\xFour - 1.5*\cw}, \yStore) {$\top$};
  \node[cellT]       at ({\xFour - 0.5*\cw}, \yStore) {$\top$};
  \node[cellT]       at ({\xFour + 0.5*\cw}, \yStore) {$\top$};
  \node[cellF] (s4D) at ({\xFour + 1.5*\cw}, \yStore) {$\bot$};
  \draw[update] (s4D.south west) rectangle (s4D.north east);
  \coordinate (s4_bot) at (\xFour, {\yStore - 0.35});
  \coordinate (s4_top) at (\xFour, {\yStore + 1});

  \draw[storearr] (cv_top) -- (s1_bot)
    node[storelabel, midway] {$\langle 00 \triangleleft \btrue\rangle$};
  \draw[storearr] (a_top) -- (s2_bot)
    node[storelabel, midway] {$\langle 10 \triangleleft \bfalse\rangle$};
  \draw[storearr] (a_top) -- (s3_bot)
    node[storelabel, midway] {$\langle 01 \triangleleft \btrue\rangle$};
  \draw[storearr] (cw_top) -- (s4_bot)
    node[storelabel, midway] {$\langle 11 \triangleleft \bfalse\rangle$};


  \coordinate (apex12) at ({(\xOne + \xTwo)/2}, \yEq);
  \coordinate (apex34) at ({(\xThree + \xFour)/2}, \yEq);

  \draw[eqarr] (s1_top) -- (apex12);
  \draw[eqarr] (s2_top) -- (apex12);
  \node[font=\Large, anchor=south] at (apex12) {$\approx$};

  \draw[eqarr] (s3_top) -- (apex34);
  \draw[eqarr] (s4_top) -- (apex34);
  \node[font=\Large, anchor=south] at (apex34) {$\approx$};
\end{tikzpicture}

%% file: rules/non_ext.tex
  \begin{tabular}{c@{\hskip 1em}c}
    \ruleinterp
    \(
    \inferrule{
      \I' \models \assertions \cup \{\teq{\aread{\awrite{a}{i}{u}}{i}}{u} \;|\; \awrite{a}{i}{u} \in \terms{\assertions}\}\\
      \I = \I_0
    }{
      \I \coloneq \I'
     }
    \)
    \\[6ex]
    \ruleinitr
    \(
    \inferrule{
      \aread{a}{i} \in \terms{\assertions}
    }{
      \propath{a}{\aread{a}{i}} \coloneq (\btrue,a)
     }
    \)
    \\[6ex]
    \ruleinitw
    \(
    \inferrule{
      \awrite{a}{i}{u} \in \terms{\assertions}\\
    }{
      \propath{\awrite{a}{i}{u}}{\aread{\awrite{a}{i}{u}}{i}} \coloneq (\btrue,\awrite{a}{i}{u})
     }
    \)
    \\[6ex]
    \rulerowd
    \(
    \inferrule{
      \I \models \tneq{i}{j}\\
      \propath{\awrite{a}{j}{u}}{\aread{b}{i}} \ne \emptypair\\
      \propath{a}{\aread{b}{i}} = \emptypair
    }{
      \propath{a}{\aread{b}{i}} \coloneq (\tneq{i}{j},\awrite{a}{j}{u})
     }
    \)
    \\[6ex]
    \rulecongr
    \(
    \inferrule{
      \I \models \teq{i}{k}\\
      \propath{a}{\aread{b}{i}} \ne \emptypair\\
      \propath{a}{\aread{c}{k}} \ne \emptypair\\
      \I \models \tneq{\aread{b}{i}}{\aread{c}{k}}\\
    }{
      \assertions \coloneq \assertions,
      (\reasonof{a}{\aread{b}{i}} \wedge \reasonof{a}{\aread{c}{k}} \wedge \teq{i}{k} \implies \teq{\aread{b}{i}}{\aread{c}{k}})\\
      (\I,\propathf) \coloneq (\I_0,\propathf_0)
     }
    \)
    \\[6ex]
    \ruleconf
    \(
    \inferrule{
      \assertions \cup \{\teq{\aread{\awrite{a}{i}{u}}{i}}{u} \;|\; \awrite{a}{i}{u} \in \terms{\assertions}\} \models \bot
    }{
      \unsat
     }
    \)
\end{tabular}

%% file: rules/ext.tex
\begin{tabular}{c@{\hskip 1em}c}
  \rulerowu
  \(
  \inferrule{
    \I \models \tneq{i}{j}\\
    \awrite{a}{j}{u} \in \terms{\assertions}\\
    \propath{a}{\aread{b}{i}} \ne \emptypair\\
    \propath{\awrite{a}{j}{u}}{\aread{b}{i}} = \emptypair
  }{
    \propath{\awrite{a}{j}{u}}{\aread{b}{i}} \coloneq (\tneq{i}{j},a)
  }
  \)
  \\[6ex]
  \ruleeqr
  \(
  \inferrule{
    \I \models \teq{a}{c}\\
    a,c \in \termsa{\assertions}\\
    \teq{a}{c} \in \terms{\assertions}\\
    \propath{a}{\aread{b}{i}} \ne \emptypair\\
    \propath{c}{\aread{b}{i}} = \emptypair
  }{
    \propath{c}{\aread{b}{i}} \coloneq (\teq{a}{c},a)
   }
  \)
  \\[6ex]
  \ruleeql
  \(
  \inferrule{
    \I \models \teq{a}{c}\\
    a,c \in \termsa{\assertions}\\
    \teq{a}{c} \in \terms{\assertions}\\
    \propath{c}{\aread{b}{i}} \ne \emptypair\\
    \propath{a}{\aread{b}{i}} = \emptypair
  }{
    \propath{a}{\aread{b}{i}} \coloneq (\teq{a}{c},c)
   }
  \)
  \\[6ex]
  \rulediseq
  \(
  \inferrule{
    \I \models \tneq{a}{c}\\
    a,c \in \termsa{\assertions}\\
    \teq{a}{c} \in \terms{\assertions}\\
    \adiff{a}{c} \not\in \terms{\assertions}
  }{
    \assertions \coloneq \assertions, (\tneq{a}{c} \implies \tneq{\aread{a}{\adiff{a}{c}}}{\aread{c}{\adiff{a}{c}}})\\
    (\I,\propathf) \coloneq (\I_0,\propathf_0)
   }
  \)
\end{tabular}

%% file: rules/cnon_ext.tex
\begin{tabular}{c@{\hskip 1em}c}
  \ruleroc
  \(
  \inferrule{
    \propath{\aconst{v}}{\aread{b}{i}} \ne \emptypair\\
    \I \models \tneq{\aread{b}{i}}{v}
  }{
    \assertions \coloneq \assertions, \reasonof{\aconst{v}}{\aread{b}{i}} \implies \teq{\aread{b}{i}}{v}\\
    (\I,\propathf) \coloneq (\I_0,\propathf_0)
   }
  \)
\end{tabular}

%% file: rules/cext.tex
\begin{tabular}{c@{\hskip 1em}c}
  \ruleinitc
  \(
  \inferrule{
    \aconst{v} \in \terms{\assertions}
  }{
    \propath{\aconst{v}}{\aconst{v}} \coloneq (\top,\aconst{v})\\
   }
  \)
  \\[6ex]
  \rulecowd
  \(
  \inferrule{
    \propath{\awrite{a}{j}{u}}{\aconst{v}} \ne \emptypair\\
    \propath{a}{\aconst{v}} = \emptypair\\
    \I \models \bexists{i : \indexsort}{\hspace{-3.2em}\bigwedge_{k \in \indicesof{\awrite{a}{j}{u}}{\aconst{v}}\cup\{j\}}\hspace{-3.0em}\tneq{i}{k}} \\
  }{
    \propath{a}{\aconst{v}} \coloneq (\btrue,\awrite{a}{j}{u})
   }
  \)
  \\[6ex]
  \rulecowu
  \(
  \inferrule{
    \propath{a}{\aconst{v}}\!\ne\!\emptypair\\
    \propath{\awrite{a}{j}{u}}{\aconst{v}}\!=\!\emptypair\\
    \awrite{a}{j}{u} \in \terms{\assertions}\\
    \I \models \bexists{i : \indexsort}{\hspace{-2.3em}\bigwedge_{k \in \indicesof{a}{\aconst{v}}\cup\{j\}}\hspace{-2.2em}\tneq{i}{k}} \\
  }{
    \propath{\awrite{a}{j}{u}}{\aconst{v}} \coloneq (\btrue,a)
   }
  \)
  \\[6ex]
  \ruleceqr
  \(
  \inferrule{
    \I \models \teq{a}{c}\\
    a,c \in \termsa{\assertions}\\
    \teq{a}{c} \in \terms{\assertions}\\
    \propath{a}{\aconst{v}} \ne \emptypair\\
    \propath{c}{\aconst{v}} = \emptypair
  }{
    \propath{c}{\aconst{v}} \coloneq (\teq{a}{c},a)
   }
  \)
  \\[6ex]
  \ruleceql
  \(
  \inferrule{
    \I \models \teq{a}{c}\\
    a,c \in \termsa{\assertions}\\
    \teq{a}{c} \in \terms{\assertions}\\
    \propath{c}{\aconst{v}} \ne \emptypair\\
    \propath{a}{\aconst{v}} = \emptypair
  }{
    \propath{a}{\aconst{v}} \coloneq (\teq{a}{c},c)
   }
  \)
  \\[6ex]
  \rulecongc
  \(
  \inferrule{
    \propath{a}{\aconst{v}} \ne \emptypair\\
    \propath{a}{\aconst{w}} \ne \emptypair\\
    \I \models \tneq{v}{w}\\
    \I \models \bexists{i : \indexsort}{\hspace{-3.4em}\bigwedge_{k \in \indicesof{a}{\aconst{v}}\cup\indicesof{a}{\aconst{w}}}\hspace{-3em}\tneq{i}{k}}
  }{
    \assertions \coloneq \assertions, (\reasonof{a}{\aconst{v}}\wedge\reasonof{a}{\aconst{w}}\wedge
    \bexists{i : \indexsort}{\hspace{-3.4em}\bigwedge_{k \in \indicesof{a}{\aconst{v}}\cup\indicesof{a}{\aconst{w}}}\hspace{-3em}\tneq{i}{k})}
    \implies \teq{v}{w}\\
  (\I,\propathf) \coloneq (\I_0,\propathf_0)
   }
  \)
\end{tabular}

%% file: table/results.tex
\newcommand{\wrongnum}[1]{\textcolor{gray}{\textsuperscript{*}#1}}

\begin{tabular}{l@{\hskip 3em}l@{\hskip 3em}l@{\hskip 2em}r@{\hskip 2em}r@{\hskip 2em}r@{\hskip 2em}r@{\hskip 2em}r}
\toprule
\multicolumn{2}{l}{\textbf{Benchmarks}}
& \textbf{Solver}
& \textbf{Solved}
& \textbf{Sat}
& \textbf{Unsat}
& \textbf{Unsolved}
& \textbf{T \scriptsize{[s]}}
\\
\midrule


\multicolumn{2}{l}{\multirow{5}{*}{\parbox{7em}{\textbf{\benchsmtlib}\\(5,112/5,112)\vspace{3ex}}}}
                     & \textbf{\configbzlaconst} & 3,545 & 3,299 & 246 & 1,567 & 1,897,536.4 \\
\multicolumn{2}{l}{} & \configziii               &   846 &   698 & 148 & 4,266 & 5,125,846.4 \\
\multicolumn{2}{l}{} & \configbzlabase           &   717 &   700 &  17 & 4,395 & 5,274,616.8 \\
\multicolumn{2}{l}{} & \configcvcv               &   276 &   184 &  92 & 4,836 & 5,803,208.3 \\

\midrule


\multicolumn{2}{l}{\multirow{5}{*}{\parbox{7em}{\textbf{\benchmbqi}\\(4,147/4,147)\vspace{3ex}}}}
                     & \textbf{\configbzlaconst} & 4,147 & 2,585 & 1,562 &     0 &        47.3 \\
\multicolumn{2}{l}{} & \configziii               & 4,126 & 2,564 & 1,562 &    21 &    30,479.6 \\
\multicolumn{2}{l}{} & \configmsat               & 4,034 & 2,522 & 1,512 &   113 &   135,663.9 \\
\multicolumn{2}{l}{} & \configcvcv               & 2,972 & 2,103 &   869 & 1,175 & 1,410,552.5 \\
\multicolumn{2}{l}{} & \configbzlabase           &     4 &     4 &     0 & 4,143 & 4,971,600.0 \\

\midrule


\multicolumn{2}{l}{\multirow{5}{*}{\parbox{7em}{\textbf{\benchpononb}\\(41/2,838)\vspace{3ex}}}}
                     & \textbf{\configbzlaconst} & 2,838 & 2,173 & 665 &     0 & 19,368.6 \\
\multicolumn{2}{l}{} & \configziii               & 2,142 & 1,644 & 498 &   696 & 47,133.4 \\
\multicolumn{2}{l}{} & \configcvcv               & 1,230 &   945 & 285 & 1,608 & 48,651.5 \\
\multicolumn{2}{l}{} & \configmsat               &   883 &   682 & 201 & 1,955 & 46,967.0 \\
\multicolumn{2}{l}{} & \configbzlabase           &   372 &    15 & 357 & 2,466 & 49,200.0 \\

\midrule


\multicolumn{2}{l}{\multirow{5}{*}{\parbox{7em}{\textbf{\benchhevm}\\(34/34)\vspace{3ex}}}}
                     & \textbf{\configbzlaconst} & 34 & 1 & 33 & 0 &     0.3 \\
\multicolumn{2}{l}{} & \configbzlabase           & 34 & 1 & 33 & 0 &     2.6 \\
\multicolumn{2}{l}{} & \configmsat               & 34 & 1 & 33 & 0 &    53.2 \\
\multicolumn{2}{l}{} & \configcvcv               & 34 & 1 & 33 & 0 &    75.9 \\
\multicolumn{2}{l}{} & \configziii               & 32 & 0 & 32 & 2 & 2,400.7 \\

\midrule


\multicolumn{2}{l}{\multirow{5}{*}{\parbox{3.5em}{\textbf{\benchcraftc}\\(205/205)}}}
& \textbf{\configbzlaconst}  & 172 & 100 & 72 &  33 &  46,752.1 \\
&& \textcolor{gray}{\configziii} & \wrongnum{153} & \wrongnum{115} & \wrongnum{38} &  \textcolor{gray}{52} & \textcolor{gray}{ 67,838.4} \\
&& \textcolor{gray}{\configmsat} & \wrongnum{107} &  \wrongnum{13} & \wrongnum{94} &  \textcolor{gray}{98} & \textcolor{gray}{119,363.6} \\
&& \configcvcv                   &   0 &   0 &  0 & 205 & 246,000.0 \\
&& \configbzlabase               &   0 &   0 &  0 & 205 & 246,000.0 \\


\cmidrule{3-8}
\multicolumn{2}{l}{\multirow{5}{*}{\parbox{3.5em}{\textbf{\benchcraftq}\\(205/205)}}}
& \textbf{\configbzlaconst}  & 117 &  89 & 28 &  88 & 109,776.4 \\
&& \configbzlabase           & 113 &  85 & 28 &  92 & 116,918.9 \\
&& \configcvcv               &  42 &  27 & 15 & 163 & 195,601.3 \\
&& \configziii               &  28 &  16 & 12 & 177 & 212,833.4 \\
\bottomrule
\end{tabular}

%% file: references.bib
@misc{smtlib25noninc,
  author       = {Mathias Preiner and
                  Hans{-}J{\"{o}}rg Schurr and
                  Clark W. Barrett and
                  Pascal Fontaine and
                  Aina Niemetz and
                  Cesare Tinelli},
  title        = {{SMT-LIB} release 2025 (non-incremental benchmarks)},
  publisher    = {Zenodo}, year         = {2025},
  month        = aug,
  doi          = {10.5281/ZENODO.15493089},
  bibsource    = {dblp computer science bibliography, https://dblp.org}
}

@misc{smtlib25inc,
  author       = {Mathias Preiner and
                  Hans{-}J{\"{o}}rg Schurr and
                  Clark W. Barrett and
                  Pascal Fontaine and
                  Aina Niemetz and
                  Cesare Tinelli},
  title        = {{SMT-LIB} release 2025 (incremental benchmarks)},
  publisher    = {Zenodo},
  year         = {2025},
  month        = may,
  doi          = {10.5281/ZENODO.15493095},
  bibsource    = {dblp computer science bibliography, https://dblp.org}
}

@inproceedings{bitwuzla,
  author       = {Aina Niemetz and
                  Mathias Preiner},
  editor       = {Constantin Enea and
                  Akash Lal},
  title        = {Bitwuzla},
  booktitle    = {Computer Aided Verification - 35th International Conference, {CAV}
                  2023, Paris, France, July 17-22, 2023, Proceedings, Part {II}},
  series       = {Lecture Notes in Computer Science},
  volume       = {13965},
  pages        = {3--17},
  publisher    = {Springer},
  year         = {2023},
  doi          = {10.1007/978-3-031-37703-7\_1},
  bibsource    = {dblp computer science bibliography, https://dblp.org}
}

@misc{bitwuzla-gh,
  title = {{Bitwuzla on GitHub}},
  year = {2026},
  howpublished = {\url{https://github.com/bitwuzla/bitwuzla/}}
}

@inproceedings{btor2,
  author       = {Aina Niemetz and
                  Mathias Preiner and
                  Clifford Wolf and
                  Armin Biere},
  editor       = {Hana Chockler and
                  Georg Weissenbacher},
  title        = {{Btor2} , {BtorMC} and {Boolector} 3.0},
  booktitle    = {Computer Aided Verification - 30th International Conference, {CAV}
                  2018, Held as Part of the Federated Logic Conference, FloC 2018, Oxford,
                  UK, July 14-17, 2018, Proceedings, Part {I}},
  series       = {Lecture Notes in Computer Science},
  volume       = {10981},
  pages        = {587--595},
  publisher    = {Springer},
  year         = {2018},
  doi          = {10.1007/978-3-319-96145-3\_32},
  bibsource    = {dblp computer science bibliography, https://dblp.org}
}

@inproceedings{cvc5,
  author       = {Haniel Barbosa and
                  Clark W. Barrett and
                  Martin Brain and
                  Gereon Kremer and
                  Hanna Lachnitt and
                  Makai Mann and
                  Abdalrhman Mohamed and
                  Mudathir Mohamed and
                  Aina Niemetz and
                  Andres N{\"{o}}tzli and
                  Alex Ozdemir and
                  Mathias Preiner and
                  Andrew Reynolds and
                  Ying Sheng and
                  Cesare Tinelli and
                  Yoni Zohar},
  editor       = {Dana Fisman and
                  Grigore Rosu},
  title        = {cvc5: {A} Versatile and Industrial-Strength {SMT} Solver},
  booktitle    = {Tools and Algorithms for the Construction and Analysis of Systems
                  - 28th International Conference, {TACAS} 2022, Held as Part of the
                  European Joint Conferences on Theory and Practice of Software, {ETAPS}
                  2022, Munich, Germany, April 2-7, 2022, Proceedings, Part {I}},
  series       = {Lecture Notes in Computer Science},
  volume       = {13243},
  pages        = {415--442},
  publisher    = {Springer},
  year         = {2022},
  doi          = {10.1007/978-3-030-99524-9\_24},
  bibsource    = {dblp computer science bibliography, https://dblp.org}
}

@inproceedings{smtinterpol,
  author       = {J{\"{u}}rgen Christ and
                  Jochen Hoenicke and
                  Alexander Nutz},
  editor       = {Alastair F. Donaldson and
                  David Parker},
  title        = {{SMTInterpol}: An Interpolating {SMT} Solver},
  booktitle    = {Model Checking Software - 19th International Workshop, {SPIN} 2012,
                  Oxford, UK, July 23-24, 2012. Proceedings},
  series       = {Lecture Notes in Computer Science},
  volume       = {7385},
  pages        = {248--254},
  publisher    = {Springer},
  year         = {2012},
  doi          = {10.1007/978-3-642-31759-0\_19},
  bibsource    = {dblp computer science bibliography, https://dblp.org}
}

@inproceedings{z3,
  author       = {Leonardo Mendon{\c{c}}a de Moura and
                  Nikolaj S. Bj{\o}rner},
  editor       = {C. R. Ramakrishnan and
                  Jakob Rehof},
  title        = {{Z3:} An Efficient {SMT} Solver},
  booktitle    = {Tools and Algorithms for the Construction and Analysis of Systems,
                  14th International Conference, {TACAS} 2008, Held as Part of the Joint
                  European Conferences on Theory and Practice of Software, {ETAPS} 2008,
                  Budapest, Hungary, March 29-April 6, 2008. Proceedings},
  series       = {Lecture Notes in Computer Science},
  volume       = {4963},
  pages        = {337--340},
  publisher    = {Springer},
  year         = {2008},
  doi          = {10.1007/978-3-540-78800-3\_24},
  bibsource    = {dblp computer science bibliography, https://dblp.org}
}

@misc{z3-gh,
  title = {{Z3 on GitHub}},
  year = {2026},
  howpublished = {\url{https://github.com/Z3Prover/z3/}}
}

@inproceedings{pono,
  author       = {Makai Mann and
                  Ahmed Irfan and
                  Florian Lonsing and
                  Yahan Yang and
                  Hongce Zhang and
                  Kristopher Brown and
                  Aarti Gupta and
                  Clark W. Barrett},
  editor       = {Alexandra Silva and
                  K. Rustan M. Leino},
  title        = {Pono: {A} Flexible and Extensible SMT-Based Model Checker},
  booktitle    = {Computer Aided Verification - 33rd International Conference, {CAV}
                  2021, Virtual Event, July 20-23, 2021, Proceedings, Part {II}},
  series       = {Lecture Notes in Computer Science},
  volume       = {12760},
  pages        = {461--474},
  publisher    = {Springer},
  year         = {2021},
  doi          = {10.1007/978-3-030-81688-9\_22},
  bibsource    = {dblp computer science bibliography, https://dblp.org}
}

@misc{hwmcc19,
  title = {{Hardware Model Checking Competition 2019}},
  year = {2019},
  howpublished = {\url{https://fmv.jku.at/hwmcc19/}}
}

@misc{hwmcc20,
  title = {{Hardware Model Checking Competition 2020}},
  year = {2020},
  howpublished = {\url{https://hwmcc.github.io/2020/}}
}

@inproceedings{hwmcc24,
  author       = {Armin Biere and
                  Nils Froleyks and
                  Mathias Preiner},
  editor       = {Nina Narodytska and
                  Philipp R{\"{u}}mmer},
  title        = {Hardware Model Checking Competition 2024},
  booktitle    = {Formal Methods in Computer-Aided Design, {FMCAD} 2024, Prague, Czech
                  Republic, October 15-18, 2024},
  pages        = {1},
  publisher    = {{IEEE}},
  year         = {2024},
  doi          = {10.34727/2024/ISBN.978-3-85448-065-5\_6},
  bibsource    = {dblp computer science bibliography, https://dblp.org}
}

@inproceedings{hwmcc25,
  author       = {Armin Biere and
                  Nils Froleyks and
                  Mathias Preiner},
  editor       = {Ahmed Irfan and
                  Daniela Kaufmann},
  title        = {Hardware Model Checking Competition 2025},
  booktitle    = {Proceedings of the 25th Conference on Formal Methods in Computer-Aided
                  Design, {FMCAD} 2025, Menlo Park, CA, USA, October 6-10, 2025},
  publisher    = {{TU} Wien Academic Press},
  year         = {2025},
  doi          = {10.34727/2025/ISBN.978-3-85448-084-6\_6},
  bibsource    = {dblp computer science bibliography, https://dblp.org}
}

@inproceedings{hevm,
  author       = {Dxo and
                  Mate Soos and
                  Zoe Paraskevopoulou and
                  Martin Lundfall and
                  Mikael Brockman},
  editor       = {Arie Gurfinkel and
                  Vijay Ganesh},
  title        = {Hevm, a Fast Symbolic Execution Framework for {EVM} Bytecode},
  booktitle    = {Computer Aided Verification - 36th International Conference, {CAV}
                  2024, Montreal, QC, Canada, July 24-27, 2024, Proceedings, Part {I}},
  series       = {Lecture Notes in Computer Science},
  volume       = {14681},
  pages        = {453--465},
  publisher    = {Springer},
  year         = {2024},
  doi          = {10.1007/978-3-031-65627-9\_22},
  bibsource    = {dblp computer science bibliography, https://dblp.org}
}

@misc{hevmbench,
title = {hevm Symbolic Execution Engine {SMT} Queries},
url = {https://github.com/SMT-LIB/pending-benchmarks/commit/0a87f2c7581a3df7ad0f2d7b57621b6ecf4637ad},
year = {2026}
}

@phdthesis{preiner-phd,
  author = {Mathias Preiner},
  title  = {Lambdas, Arrays and Quantifiers},
  year   = {2017},
  series = {Dissertation Technische Wissenschaften},
  school = {Informatik, Johannes Kepler University Linz},
  url = {https://resolver.obvsg.at/urn:nbn:at:at-ubl:1-14897}
}

@TECHREPORT{BarFT-RR-25,
  author      = {Clark Barrett and Pascal Fontaine and Cesare Tinelli},
  title       = {{The SMT-LIB Standard: Version 2.7}},
  institution = {Department of Computer Science, The University of Iowa},
  year        = 2025,
  note        = {Available at \url{http://smt-lib.org}}
}

@incollection{Man-MSL-93,
  Address = {New York, NY, USA},
  Author = {Manzano, Mar\'{i}a},
  Booktitle = {Many-sorted logic and its applications},
  note = {{ISBN}  0-471-93485-2},
  Pages = {3--86},
  Publisher = {John Wiley \& Sons, Inc.},
  Title = {Introduction to many-sorted logic},
  Year = {1993},
  url = {https://dl.acm.org/doi/10.5555/165446.165450}
}

@book{EndertonLogic,
  author       = {Herbert B. Enderton},
  title        = {A mathematical introduction to logic},
  publisher    = {Academic Press},
  year         = {1972},
  note = {{ISBN} 978-0-12-238450-9},
  bibsource    = {dblp computer science bibliography, https://dblp.org}
}

@inproceedings{McCarthy62,
  author       = {John McCarthy},
  title        = {Towards a Mathematical Science of Computation},
  booktitle    = {Information Processing, Proceedings of the 2nd {IFIP} Congress 1962,
                  Munich, Germany, August 27 - September 1, 1962},
  pages        = {21--28},
  publisher    = {North-Holland},
  year         = {1962},
  bibsource    = {dblp computer science bibliography, https://dblp.org}
}

@inproceedings{StumpBDL01,
  author       = {Aaron Stump and
                  Clark W. Barrett and
                  David L. Dill and
                  Jeremy R. Levitt},
  title        = {A Decision Procedure for an Extensional Theory of Arrays},
  booktitle    = {16th Annual {IEEE} Symposium on Logic in Computer Science, Boston,
                  Massachusetts, USA, June 16-19, 2001, Proceedings},
  pages        = {29--37},
  publisher    = {{IEEE} Computer Society},
  year         = {2001},
  doi          = {10.1109/LICS.2001.932480},
  bibsource    = {dblp computer science bibliography, https://dblp.org}
}

@inproceedings{HoenickeS19,
  author       = {Jochen Hoenicke and
                  Tanja Schindler},
  editor       = {Constantin Enea and
                  Ruzica Piskac},
  title        = {Solving and Interpolating Constant Arrays Based on Weak Equivalences},
  booktitle    = {Verification, Model Checking, and Abstract Interpretation - 20th International
                  Conference, {VMCAI} 2019, Cascais, Portugal, January 13-15, 2019,
                  Proceedings},
  series       = {Lecture Notes in Computer Science},
  volume       = {11388},
  pages        = {297--317},
  publisher    = {Springer},
  year         = {2019},
  doi          = {10.1007/978-3-030-11245-5\_14},
  bibsource    = {dblp computer science bibliography, https://dblp.org}
}

@inproceedings{MouraB09,
  author       = {Leonardo Mendon{\c{c}}a de Moura and
                  Nikolaj S. Bj{\o}rner},
  title        = {Generalized, efficient array decision procedures},
  booktitle    = {Proceedings of 9th International Conference on Formal Methods in Computer-Aided
                  Design, {FMCAD} 2009, 15-18 November 2009, Austin, Texas, {USA}},
  pages        = {45--52},
  publisher    = {{IEEE}},
  year         = {2009},
  doi          = {10.1109/FMCAD.2009.5351142},
  bibsource    = {dblp computer science bibliography, https://dblp.org}
}

@article{BrummayerB09,
  author       = {Robert Brummayer and
                  Armin Biere},
  title        = {Lemmas on Demand for the Extensional Theory of Arrays},
  journal      = {J. Satisf. Boolean Model. Comput.},
  volume       = {6},
  number       = {1-3},
  pages        = {165--201},
  year         = {2009},
  doi          = {10.3233/SAT190067},
  bibsource    = {dblp computer science bibliography, https://dblp.org}
}

@inproceedings{PreinerNB13,
  author       = {Mathias Preiner and
                  Aina Niemetz and
                  Armin Biere},
  editor       = {Malay K. Ganai and
                  Alper Sen},
  title        = {Lemmas on Demand for Lambdas},
  booktitle    = {Proceedings of the Second International Workshop on Design and Implementation
                  of Formal Tools and Systems, Portland, OR, USA, October 19, 2013},
  series       = {{CEUR} Workshop Proceedings},
  volume       = {1130},
  publisher    = {CEUR-WS.org},
  year         = {2013},
  url          = {https://ceur-ws.org/Vol-1130/paper\_7.pdf},
  bibsource    = {dblp computer science bibliography, https://dblp.org}
}

@inproceedings{BarrettDS02,
  author    = {Clark W. Barrett and
               David L. Dill and
               Aaron Stump},
  editor    = {Ed Brinksma and
               Kim Guldstrand Larsen},
  title     = {Checking Satisfiability of First-Order Formulas by Incremental Translation
               to {SAT}},
  booktitle = {Computer Aided Verification, 14th International Conference, {CAV}
               2002,Copenhagen, Denmark, July 27-31, 2002, Proceedings},
  series    = {Lecture Notes in Computer Science},
  volume    = {2404},
  pages     = {236--249},
  publisher = {Springer},
  year      = {2002},
  doi       = {10.1007/3-540-45657-0\_18},
  bibsource = {dblp computer science bibliography, https://dblp.org}
}

@inproceedings{DemouraR02,
  author    = {Leonardo De Moura and
               Harald Rue{\ss}},
  title     = {Lemmas on Demand for Satisfiability Solvers},
  booktitle = {The 5th International Symposium on the Theory and Applications of Satisfiability Testing, {SAT}
               2002, Cincinnati, USA},
  documenturl = { http://leodemoura.github.io/files/sat02.pdf},
  year      = {2002},
  url       = {https://leodemoura.github.io/files/sat02.pdf}
}

@inproceedings{GeM09,
  author    = {Yeting Ge and
               Leonardo Mendon{\c{c}}a de Moura},
  editor    = {Ahmed Bouajjani and
               Oded Maler},
  title     = {Complete Instantiation for Quantified Formulas in Satisfiabiliby Modulo
               Theories},
  booktitle = {Computer Aided Verification, 21st International Conference, {CAV}
               2009, Grenoble, France, June 26 - July 2, 2009. Proceedings},
  series    = {Lecture Notes in Computer Science},
  volume    = {5643},
  pages     = {306--320},
  publisher = {Springer},
  year      = {2009},
  doi       = {10.1007/978-3-642-02658-4\_25},
  bibsource = {dblp computer science bibliography, https://dblp.org}
}

@inproceedings{NiemetzPZ24,
  author       = {Aina Niemetz and
                  Mathias Preiner and
                  Yoni Zohar},
  editor       = {Arie Gurfinkel and
                  Vijay Ganesh},
  title        = {Scalable Bit-Blasting with Abstractions},
  booktitle    = {Computer Aided Verification - 36th International Conference, {CAV}
                  2024, Montreal, QC, Canada, July 24-27, 2024, Proceedings, Part {I}},
  series       = {Lecture Notes in Computer Science},
  volume       = {14681},
  pages        = {178--200},
  publisher    = {Springer},
  year         = {2024},
  doi          = {10.1007/978-3-031-65627-9\_9},
  bibsource    = {dblp computer science bibliography, https://dblp.org}
}

@inproceedings{yices2,
  author       = {Bruno Dutertre},
  editor       = {Armin Biere and
                  Roderick Bloem},
  title        = {Yices 2.2},
  booktitle    = {Computer Aided Verification - 26th International Conference, {CAV}
                  2014, Held as Part of the Vienna Summer of Logic, {VSL} 2014, Vienna,
                  Austria, July 18-22, 2014. Proceedings},
  series       = {Lecture Notes in Computer Science},
  pages        = {737--744},
  publisher    = {Springer},
  year         = {2014},
  doi          = {10.1007/978-3-319-08867-9\_49},
  bibsource    = {dblp computer science bibliography, https://dblp.org}
}

@inproceedings{mathsat5,
  author       = {Alessandro Cimatti and
                  Alberto Griggio and
                  Bastiaan Joost Schaafsma and
                  Roberto Sebastiani},
  editor       = {Nir Piterman and
                  Scott A. Smolka},
  title        = {The MathSAT5 {SMT} Solver},
  booktitle    = {Tools and Algorithms for the Construction and Analysis of Systems
                  - 19th International Conference, {TACAS} 2013, Held as Part of the
                  European Joint Conferences on Theory and Practice of Software, {ETAPS}
                  2013, Rome, Italy, March 16-24, 2013. Proceedings},
  series       = {Lecture Notes in Computer Science},
  pages        = {93--107},
  publisher    = {Springer},
  year         = {2013},
  doi          = {10.1007/978-3-642-36742-7\_7},
  bibsource    = {dblp computer science bibliography, https://dblp.org}
}

@book{handbookmc18,
  editor       = {Edmund M. Clarke and
                  Thomas A. Henzinger and
                  Helmut Veith and
                  Roderick Bloem},
  title        = {Handbook of Model Checking},
  publisher    = {Springer},
  year         = {2018},
  doi          = {10.1007/978-3-319-10575-8},
  note         = {{ISBN} 978-3-319-10574-1},
  bibsource    = {dblp computer science bibliography, https://dblp.org}
}

@article{BaldoniCDDF18,
  author       = {Roberto Baldoni and
                  Emilio Coppa and
                  Daniele Cono D'Elia and
                  Camil Demetrescu and
                  Irene Finocchi},
  title        = {A Survey of Symbolic Execution Techniques},
  journal      = {{ACM} Comput. Surv.},
  volume       = {51},
  number       = {3},
  pages        = {50:1--50:39},
  year         = {2018},
  doi          = {10.1145/3182657},
  bibsource    = {dblp computer science bibliography, https://dblp.org}
}

@misc{arxiv,
  title={{Satisfiability Modulo Extensional Constant Arrays (Extended Version)}},
  author={Mathias Preiner and Aina Niemetz and Clark Barrett},
  year={2026},
  eprint={2605.16820},
  archivePrefix={arXiv},
  primaryClass={cs.LO},
  doi={10.48550/arXiv.2605.16820},
}
